\documentclass[11pt,reqno]{article}
\usepackage[margin=1.2in]{geometry}
\usepackage{dsfont,amssymb,amsmath,amscd,latexsym,amsthm,amsxtra,amsfonts}
\allowdisplaybreaks[4]
\usepackage{lineno}
\usepackage[all]{xy}
\usepackage{enumitem}
\usepackage[active]{srcltx}
\usepackage{tikz}
\usepackage[round]{natbib}
\usepackage{mathrsfs}
\usepackage{graphicx}
\usepackage{placeins}
\usepackage{subcaption}
\usepackage{comment}
\usepackage{mathtools}
\usepackage{cases}
\usepackage{tcolorbox}
\tcbuselibrary{most}

\usetikzlibrary{calc,arrows}
\usepackage{verbatim}
\usepackage{color}
\usepackage{epstopdf}
\usepackage[affil-it]{authblk}
\usepackage{bm}
\usepackage[title]{appendix}

\newtheorem{theorem}{Theorem}[section]
\newtheorem{assumption}[theorem]{Assumption}

\newtheorem{corollary}[theorem]{Corollary}
\newtheorem{definition}[theorem]{Definition}
\newtheorem{example}[theorem]{Example}

\newtheorem{lemma}[theorem]{Lemma}

\newtheorem{proposition}[theorem]{Proposition}
\newtheorem{remark}[theorem]{Remark}
\numberwithin{equation}{section}

\DeclareMathOperator{\essinf}{essinf}
\DeclareMathOperator*{\supp}{supp}
\newcommand{\me}{\mathrm{e}}

\newcommand{\md}{\mathrm{d}}

\newcommand{\mR}{\mathbb{R}}

\newcommand{\mS}{\mathbb{S}}

\newcommand{\mE}{\mathbb{E}}
\newcommand{\mEt}{\tilde{\mathbb{E}}}
\newcommand{\mF}{\mathbb{F}}

\renewcommand{\epsilon}{\varepsilon}

\newcommand{\F}{\mathcal{F}}

\newcommand{\barpi}{\bar{\pi}}

  \usepackage[pdfstartview=FitH, bookmarksnumbered=true,bookmarksopen=true, colorlinks=true, pdfborder={0 0 1}, citecolor=blue, linkcolor=blue,urlcolor=blue]{hyperref}
\usepackage{graphics}
\graphicspath{{figures/}}
\renewcommand{\baselinestretch}{1.2}
\usepackage[utf8]{inputenc}
\usepackage{newunicodechar}
\newunicodechar{，}{,}

\title{Equilibrium Investment with Random Risk Aversion:  (Non-)uniqueness, Optimality, and Comparative Statics
}

\author{Weilun Cheng\thanks{Department of Mathematical Sciences, Tsinghua University,  Beijing 100084, China (\url{ chengwl25@mails.tsinghua.edu.cn}).}, \ 
Zongxia Liang\thanks{Department of Mathematical Sciences, and Center for Insurance and Risk Management, School of Economics and Management, Tsinghua University， Beijing 100084, China
  (\url{liangzongxia@tsinghua.edu.cn}).}, \ 
    Sheng Wang\thanks{Department of Statistics and Actuarial Science, The University of Hong Kong, Pokfulam Road, Hong Kong, China (\url{sheng-wa15@tsinghua.org.cn}).}, \  Jianming Xia\thanks{State Key Laboratory of Mathematical Sciences, Academy of Mathematics and Systems Science, Chinese
Academy of Sciences, Beijing 100190, China (\url{xia@amss.ac.cn})}
}

\begin{document}
	
\maketitle
	
\maketitle

\begin{abstract}
 This paper studies a continuous-time portfolio selection problem under a general distribution of random risk aversion (RRA). We provide a complete characterization of all deterministic equilibrium strategies in closed form. Our results show that the structure of the solution depends crucially on the distribution of RRA: the equilibrium is unique (if exits) when the expectation of RRA is finite, whereas an infinite expectation leads either to infinitely many equilibria or to a unique trivial one (i.e. risk-free investment).
To resolve this multiplicity of equilibria, we select, among all deterministic equilibria, the one that maximizes the objective functional at the initial time. We establish a necessary and sufficient condition for the existence of such an optimal equilibrium, which is then shown to be unique and uniformly optimal.
Finally, we conduct a comparative statics. Using counterexamples based on two-point distributed RRA, we demonstrate that a larger risk aversion in the sense of first-order stochastic dominance does not necessarily lead to less risky investment. Within the two-point distribution framework, we further examine the single-crossing property of equilibrium strategies and the monotonicity of the crossing time. We show that a larger risk aversion under a stronger stochastic order---the reverse hazard rate order---always leads to less risky investment. In addition, we analyze how the convex combination of independent and identically distributed RRAs influences investment. 
\end{abstract}


\section{Introduction}\label{intro}

The classical framework for portfolio selection, pioneered by \cite{merton1969lifetime,merton1971} and \cite{samuelson1969lifetime}, is based on the assumption that an investor's risk aversion can be captured by a single known constant. In practice, however, this coefficient is very difficult for an agent to specify with certainty, and a growing body of empirical evidence suggests that it varies dynamically with market conditions, often increasing during financial crises and bear markets, as shown in studies by \cite{gordon2000preference}, \cite{chetty2006new}, and \cite{guiso2018time}. These challenges have motivated a move towards models that incorporate random risk aversion (RRA). Within this paradigm, a more economically coherent approach, as advanced by \cite{Desmettre2023}, is to aggregate the certainty equivalents over the distribution of the risk aversion parameter. 
While the aggregation in \cite{Desmettre2023} provides a consistent performance measure in monetary units, it introduces a nonlinearity in expectation, which naturally leads to the time inconsistency of the preference.

The first systematic treatment of time inconsistency was given by \cite{Strotz1955}, who proposed the notion of consistent planning, laying the foundation for the intra-personal game-theoretic framework. Building on this idea, \cite{ekeland2006being} later developed a rigorous continuous-time formalization, which in turn inspired extensive research on continuous-time control with time-inconsistent preferences. For example,  \cite{bjork2010general} (published version \cite{Bjork2017}) established a general theoretical framework by deriving an extended HJB equation and \cite*{hu2012time,Hu2017} investigated a time-inconsistent stochastic linear–quadratic control problem. For further discussions and related developments， see， e.g., \cite{yan2019time}, \cite{he2021equilibrium}, \cite{Hernandez2023}, and the references therein.

Within the Black-Scholes market model, \cite{Desmettre2023} provide a verification theorem on the intra-personal equilibrium strategies for a general RRA and, in the case of a two-point distributed RRA, characterize the equilibrium by a three-dimensional ODE system without establishing the existence and uniqueness of the solution.\footnote{This work has spurred further research extending the model to more complex settings. For example, \cite*{chen2025equilibrium} incorporate a regime-switching framework in which market dynamics and preferences co-evolve; \cite{aquino2025equilibrium} study a continuous-time portfolio choice problem where state-dependent preferences are determined by an exogenous factor evolving as an Itô diffusion; \cite{WANG2025103140} study the management of the DC pension plan, providing semi-explicit solutions for a two-point distributed RRA; \cite*{he2025dynamic} establish the existence and uniqueness of the solution to an infinite-dimensional Riccati equation which characterizes the equilibrium strategies under Heston's SV model for a general bounded
RRA. For an investigation on the corresponding pre-commitment problem, see, e.g., \citet{xia2024optimal}.} 
This technical gap is filled by \cite*{liang2025short}.   For a general RRA, \cite*{liang2025short} characterize 
intra-personal equilibrium strategies using an integral equation. They establish the existence and uniqueness of the solution to the integral equation under the assumptions that the variance of the RRA is finite and that the function $h$ defined in (\ref{eq:h}) below is bounded. 

This paper aims to provide deeper insight into the portfolio selection problem with RRA within the simplest modeling framework, in line with \cite{Desmettre2023} and \cite*{liang2025short}. Our main contributions are summarized below.
\begin{itemize}
\item \textbf{Complete characterization of all deterministic equilibria in closed form.} 
We relax the assumptions in the existing literature and derive, in closed form, all deterministic equilibrium strategies. An important finding is that there may be infinitely many equilibria when the expectation of the RRA is infinite. The uniqueness of the equilibrium strategy for the continuous-time time-inconsistent control problem is a very challenging problem (see, e.g.， \citet*[Section 7]{hjz2021}). This paper sheds light on this issue by demonstrating that equilibrium strategies, even within the class of deterministic strategies, may lack uniqueness. Furthermore, we provide a necessary and sufficient condition for the uniqueness of the deterministic equilibrium.

\item \textbf{Resolution of equilibrium multiplicity via the optimality criterion.}
A question then arises: which one should the agent choose from these infinitely many equilibrium strategies? 
We suggest that the agent should maximize the objective functional at the initial time over all available equilibria. This leads to a discussion of the so-called \emph{optimal} equilibria. 
To our best knowledge, there is only one existing paper, \cite*{wei2026time}, on optimal equilibria for a time-inconsistent \emph{control} problem, and there are three papers, \citet{hz2019,hz2020mf} and \citet{hw2021}, on optimal equilibria for time-inconsistent \emph{stopping} problems. 
We give a necessary and sufficient condition for the existence of optimal equilibria and observe that the optimal equilibrium, if it exists, is unique and is also \emph{uniformly} optimal. Moreover, if the market price of risk remains non-zero near the terminal time, the optimal equilibrium is also \emph{uniformly strictly} optimal.

\item \textbf{Comparative statics on equilibrium strategies.}
We investigate how the distribution of RRA affects the equilibrium strategies.  Intuitively, one would expect a “comparative property”: a “larger” RRA should lead to less risky investment. However, through counterexamples involving two-point distributed RRA, we demonstrate that “largeness” in the sense of first-order stochastic dominance is \emph{insufficient} to guarantee this property. Within the same two-point distribution framework, we further examine the single-crossing property of equilibrium strategies and the monotonicity of the crossing time with respect to (w.r.t.) the probability parameter. In contrast, a stronger stochastic order—reverse hazard rate dominance—is shown to be \emph{sufficient} for the comparative property to hold. This yields an analogy with \cite{Borell2007} and \cite{xia2011risk}, who conducted comparative statics on optimal portfolios w.r.t. risk aversion in the context of expected utility maximization with general utility functions under the Black–Scholes market model. (The corresponding analysis for a static model with one risk-free asset and one risky asset dates back to \cite{Arrow1970} and \cite{Pratt1964:RiskAversion}.) Finally, we show that a convex combination of independent and identically distributed (i.i.d.) RRAs results in less risky investment.
\end{itemize}

The rest of the paper is organized as follows. Section \ref{sec:time-inconsistent} introduces the model and derives the integral equation that characterizes the equilibrium. Section \ref{Equilibrium Analysis and Characterization of Solutions} analyzes both finite- and infinite-expectation cases, offering closed-form solutions for all equilibria. Section \ref{Optimal Equilibrium Selection} resolves equilibrium multiplicity by proposing an optimality-based selection criterion, and establishes necessary and sufficient conditions for the existence, uniqueness, and uniform (strict) optimality of optimal equilibria. Section \ref{Comparative Statics of Equilibrium Strategies} performs a comparative statics of equilibrium strategies w.r.t. the distribution and the aggregation of RRA, identifying conditions under which “larger” risk aversion leads to less risky investment.

\section{Problem Formulation}\label{sec:time-inconsistent}

Let $T>0$ be a finite time horizon and $(\Omega, \mathcal{F}, \mathbb{F}, \mathbb{P})$ a filtered complete probability space, where $\mathbb{F}=\{\mathcal{F}_{t}:\ 0\le t\le T\}$ is the augmented natural filtration generated by a $d$-dimensional standard Brownian motion $\{B(t):\ 0\le t\le T\}$ and $\mathcal{F}=\mathcal{F}_{T}$. The market consists of one risk-free asset (bond) and $d$ risky assets (stocks). For simplicity, we assume that the interest rate of the bond is zero. The dynamics of the stock price processes $S_{i}$, $i=1,\cdots,d$, is given by
$$ dS_{i}(t)=S_{i}(t)[\mu_{i}(t)dt+\sigma_{i}(t)dB(t)],\quad t\in[0,T], $$
where the market coefficients $\mu:[0,T]\rightarrow\mathbb{R}^{d}$ and $\sigma:[0,T]\rightarrow\mathbb{R}^{d\times d}$ are right-continuous and deterministic (each $\sigma_{i}$ denotes the $i$-th row of $\sigma$). We always assume that $$\int_{0}^{T}|\mu(t)|dt+\sum_{i=1}^{d}\int_{0}^{T}|\sigma_{i}(t)|^{2}dt<\infty$$ and that $\sigma(t)$ is invertible for every $t\in[0,T]$. The market price of risk is $\lambda(t)\triangleq(\sigma(t))^{-1}\mu(t)$, $t\in[0,T]$. We also assume that $\int_0^T|\lambda(t)|^2\md t<\infty$.

For any $m\geq 1$ and $\mS \subset \mR^m$, $L^0(\mF,\mS)$ is the space consisting of $\mS$-valued, $\mF$-progressively measurable processes. For each $t\in[0,T]$, $p\in[1,\infty]$, $L^p(\F_t,\mS)$ is the set of all $L^p$-integrable, $\mS$-valued, and $\F_t$-measurable random variables. 

A trading strategy is a process $\pi=\{\pi_t:\ t\in[0,T)\}\in L^0(\mF,\mR^d)$ such that $$\int_0^T|\pi^\top_t\mu(t)|\md t+\int_0^T|\sigma^\top(t)\pi_t|^2\md t<\infty\text{ a.s.,}$$ where $\pi_t$ is the vector of portfolio weights at time $t$. 
For any $t\in[0,T)$ and any initial wealth level $w>0$, let $W^{t,w,\pi}=\{W^{t,w,\pi}_s:\ t\le s\le T\}$ denote the self-financing wealth process starting from $w$ at time $t$ and driven by a trading strategy $\pi$. It evolves according to the stochastic differential equation (SDE):
\begin{equation}\label{wealthdynamic}
	\left\{
	\begin{aligned}
		&\md W^{t,w,\pi}_s=W^{t,w,\pi}_s\pi^{\top}_s\mu(s)\md s+W^{t,w,\pi}_s \pi^{\top}_s\sigma(s)\md B(s),\quad s\in[t,T],\\	
		&W^{t,w,\pi}_t=w.
	\end{aligned}
	\right.
\end{equation}

Consider the constant relative risk aversion (CRRA) utility functions:
\begin{align*}
u^{\gamma}(x) = 
\begin{cases} 
    \frac{x^{1-\gamma}}{1-\gamma}, & \gamma \ge 0, \gamma \neq 1, x > 0, \\
    \log x, & \gamma = 1, x > 0.
\end{cases}
\end{align*}
Under the expected utility theory, for each CRRA utility function $u^\gamma$, at time $t$, the certainty equivalent of $W^{t,w,\pi}_T$ is
\(\left(u^\gamma\right)^{-1}\left(\mE_{t}\left[u^\gamma\left(W^{t,w,\pi}_T\right)\right]\right)\),
where $\mE_t$ denotes the conditional expectation w.r.t. $\mathcal{F}_t$.
Following \cite{Desmettre2023}, we assume that the objective functional of the agent at time $t$ is  
the weighted average of the certainty equivalents: 
\begin{equation}\label{J:rra}
J\left(\pi;t,w\right)=\int_{[0,\infty)}\left(u^\gamma\right)^{-1}\left(\mE_{t}\left[u^\gamma\left(W^{t,w,\pi}_T\right)\right]\right) \md F_{\boldsymbol{R}}(\gamma),
\end{equation}
where $F_{\boldsymbol{R}}$ is the probability distribution function of the RRA $\boldsymbol{R}$, a non-negative random variable defined on another independent probability space $(\tilde{\Omega}, \tilde{\mathcal{F}}, \tilde{\mathbb{P}})$. We always assume $F_{\boldsymbol{R}}(0)<1$, which amounts to $\tilde{\mathbb{P}}(\boldsymbol{R}>0)>0$.\footnote{We allow $\tilde{\mathbb{P}}(\boldsymbol{R}=0)>0$, which accommodates some common distributions, such as the Poisson distribution discussed in Example~\ref{ex:h_calculation}. }
 We denote by $\mEt[\cdot]$ the expectation operator under $\tilde{\mathbb{P}}$.

A trading strategy $\pi$ is called \emph{admissible}  if for all $t\in[0,T)$ and $w>0$,  
 \begin{align*}
 \begin{cases}
\mE_{t}\left[\left|u^\gamma\left(W^{t,w,\pi}_T\right)\right|\right]<\infty\, \ a.s.\quad \text{ for all } \gamma \in\supp(F_{\boldsymbol{R}}),\\
\displaystyle\int_{[0,\infty)}\left|\left(u^\gamma\right)^{-1}\left(\mE_{t}\left[u^\gamma\left(W^{t,w,\pi}_T\right)\right]\right) \right|\md F_{\boldsymbol{R}}(\gamma)<\infty\,\ \  a.s.
\end{cases}
 \end{align*}
Let $\Pi$ denote the set of all admissible strategies. 

Hereafter, we  consider a fixed $\barpi\in\Pi$, which is a candidate equilibrium strategy. 
For any $t\in[0,T)$, $\epsilon\in(0,T-t)$, and $k\in  L^\infty(\F_t,\mathbb{R}^d)$,   let $\pi^{t,\epsilon,k}$  be defined by
\[\pi^{t,\epsilon,k}_s\triangleq
\left\{
\begin{aligned}
	&{\barpi_s+k}, &s\in [t,t+\epsilon),\\
	&\barpi_s, &s\notin [t,t+\epsilon).
\end{aligned}
\right.
\]
The strategy $\pi^{t,\epsilon,k}$ serves as a perturbation of $\barpi$. 
Following \cite*{hu2012time,Hu2017} and
\cite*{yan2019time}, we introduce the definition of equilibrium strategies as follows. 

\begin{definition}\label{def:equilibrium}
A strategy $\barpi$ is called an equilibrium strategy if, for any $t \in [0,T)$, $w>0$,  $k \in L^{\infty}(\mathcal{F}_t, \mathbb{R}^d)$ with $\pi^{t, \epsilon, k} \in \Pi$ for all sufficiently small $\epsilon > 0$, and  any positive sequence $\{\epsilon_n, n \ge 1\}$ with $\lim\limits_{n \to \infty} \epsilon_n = 0$, we have
\begin{align}\label{limsup:J}
\limsup_{n \to \infty} \frac{J( \pi^{t, \epsilon_n, k};t,w) - J( \barpi;t,w)}{\epsilon_n} \le 0 \quad \text{a.s.}
\end{align}
\end{definition}

\begin{remark}
In literature, the following inequality 
\begin{align*}
\limsup_{\epsilon \downarrow 0} \frac{J( \pi^{t, \epsilon, k};t,w) - J(\barpi;t,w)}{\epsilon} \le 0 \quad \text{a.s.}
\end{align*}
is usually used to define equilibrium strategies. However, the left-hand side of the inequality might be unmeasurable. To guarantee the measurability, we use the upper limit of a sequence instead of the upper limit of an uncountable net.
\end{remark}
\begin{remark}\label{rmk:j_0}
    Note that
    \begin{align*}
        J(\pi;t,w)&=w\int_{[0,\infty)}\left(u^\gamma\right)^{-1}\left(\mE_{t}\left[u^\gamma\left({W^{t,w,\pi}_T}/ w\right)\right]\right) \md F_{\boldsymbol{R}}(\gamma)=wJ_0(t,\pi),
\end{align*}
where 
\begin{align*} J_0(t,\pi)\triangleq\int_{[0,\infty)}\left(u^\gamma\right)^{-1}\left(\mE_{t}\left[u^\gamma\left({W^\pi_T/ W^\pi_t}\right)\right]\right) \md F_{\boldsymbol{R}}(\gamma)
\end{align*}
is the objective functional used in \cite*{liang2025short} and $W^\pi$ evolves according to
\begin{align*}
    \left\{
    \begin{aligned}
        &\md W^\pi_t=W^\pi_t \pi^{\top}_t\mu(t)\md t+W^\pi_t \pi^{\top}_t\sigma(t)\md B(t),\quad	t\in[0,T],\\
		&W^\pi_0=w_0>0.   
    \end{aligned}
    \right.
\end{align*}
Thus, under Definition \ref{def:equilibrium}, an equilibrium strategy for $J$ is also an equilibrium strategy for $J_0$, and vice versa. We will frequently use $J_0$ in the subsequent analysis.
\end{remark}

In our analysis, we follow the approach of \cite*{liang2025short} and focus on the equilibrium strategies given by $\bar{\pi}=(\sigma^\top)^{-1}a$, where $a$, referred to as the \emph{risk exposure vector}, is a deterministic and right-continuous function in the $L^2$-space. Let $\Pi_{d} \subset \Pi$ denote the set of all such strategies. Because $k\in L^{\infty}(\mathcal{F}_t, \mathbb{R}^d)$, the perturbed strategy $\pi^{t,\epsilon,k}$ can be random, even though $\barpi\in\Pi_d$ itself is deterministic. 
The wealth process $W^{t,w,\barpi}$ of the portfolio $\barpi$ satisfies the following SDE
\begin{align*}
    \left\{
    \begin{aligned}
        &\md W^{t,w,\barpi}_s = W^{t,w,\barpi}_s[a(s)^\top\lambda(s)\md s+a(s)^\top\md B(s)],\quad s\in[t,T],\\
        &W^{t,w,\barpi}_t=w,   
    \end{aligned}
    \right.
\end{align*}
and $|a|$ is referred to as the \emph{risk exposure magnitude} (or local volatility).

For any $\bar{\pi} = (\sigma^\top)^{-1} a \in \Pi_{d}$, we introduce the following notations for ease of analysis:
$$
v_{a}(t) \triangleq \int_{t}^{T} |a(s)|^2 ds, \quad y_{a}(t) \triangleq \int_{t}^{T} a^{\top}(s)\lambda(s)ds,\quad t\in[0,T].
$$
Because the strategy is deterministic, the conditional distribution of the relative wealth $W^{t,w,\barpi}_T/w$ depends only on the deterministic quantities $v_{a}(t)$ and $y_{a}(t)$ and we have
\begin{align}\label{the definition of J}
J_0(t,\barpi)=\tilde{\mE}\left[\me^{-\frac{1}{2}\left(\boldsymbol{R}v_a(t)-2y_a(t)\right) }\right].
\end{align}
 
The following theorem shows that $\barpi\equiv \mathbf{0}$ is an equilibrium if $\tilde{\mE}[\boldsymbol{R}]=\infty$. (This result is natural: if the RRA is very large, then the investor puts no money into the risky asset. However, in Theorem \ref{theorem3.3}(2) we will see that $\barpi\equiv \mathbf{0}$ is not necessarily the unique equilibrium.)

\begin{theorem}\label{thm:0:equi}
The risk-free investment ($\barpi\equiv \mathbf{0}$) is an equilibrium if $\tilde{\mE}[\boldsymbol{R}]=\infty$.
\end{theorem}
\begin{proof}
    Let $\barpi=\mathbf{0}$. For any $t\in[0,T)$, $\epsilon\in(0,T-t)$, and $k\in  L^\infty(\F_t,\mathbb{R}^d)$, we have $J_0(t,\barpi)=1$ and 
        \begin{align*} J_0(t,\pi^{t,\epsilon,k})=\me^{\int_t^{t+\epsilon}\mu^{\top}(s)k\md s}\int_{[0,\infty)}\me^{-\frac{\gamma}{2}\int_t^{t+\epsilon}|\sigma^{\top}(s)k|^2\md s}\md F_{\boldsymbol{R}}(\gamma)\triangleq \phi(\epsilon) \psi (\epsilon).
        \end{align*}
       Noting that
       \begin{align*}
           \lim_{\epsilon\to 0}\frac{\psi(\epsilon)-1}{\epsilon}
           &=-\frac{1}{2}|\sigma^{\top}(t)k|^2\lim_{\epsilon\to 0}\int_{[0,\infty)}\gamma\me^{-\frac{\gamma}{2}\int_t^{t+\epsilon}|\sigma^{\top}(s)k|^2\md s}\md F_{\boldsymbol{R}}(\gamma)\\
           &=-\frac{1}{2}|\sigma^{\top}(t)k|^2\int_{[0,\infty)}\gamma\md F_{\boldsymbol{R}}(\gamma)=-\frac{1}{2}|\sigma^{\top}(t)k|^2\tilde{\mE}[\boldsymbol{R}],
       \end{align*}
       where the first equality uses  the mean value theorem
       and the second the monotone convergence theorem. Thus,  for any positive sequence $\{\epsilon_n, n \ge 1\}$ satisfying $\lim\limits_{n \to \infty} \epsilon_n = 0$, we have 
       \begin{align*}
           \lim_{n\to \infty} \frac{J_0(t, \pi^{t, \epsilon_n, k}) - J_0(t, \barpi)}{\epsilon_n}&= \lim_{n\to \infty}\frac{\phi(\epsilon_n)\psi(\epsilon_n)-1}{\epsilon_n}=\lim_{n\to \infty}\frac{(\phi(\epsilon_n)-1)\psi(\epsilon_n)+\psi(\epsilon_n)-1}{\epsilon_n}\\
           &=\left(\mu^{\top}(t)k-\frac{1}{2}|\sigma^{\top}(t)k|^2\tilde{\mE}[\boldsymbol{R}]\right).
       \end{align*}   
       Therefore, $\mathbf{0}$ is an equilibrium if $\tilde{\mE}[\boldsymbol{R}]=\infty$. 
\end{proof}

The next theorem provides a necessary and sufficient condition for 
$\bar{\pi} = (\sigma^\top)^{-1} a \in \Pi_{d}$ to be an equilibrium.

\begin{theorem} \label{thm:integral:eq}
A strategy $\bar{\pi} = (\sigma^\top)^{-1} a \in \Pi_{d}$ is an equilibrium  if and only if
\begin{align}\label{eq:equilibrium:g}
a(t)=h(v_a(t))\lambda(t)
		\end{align}
holds for any $t\in[0,T)$, where\footnote{This function $h$ is slightly different from that one in \cite*{liang2025short}:  $h(x)$ there amounts to $h(x^2)$ here. Nevertheless, the relevant properties, such as the boundedness of $h$, remain equivalent under the two formulations.}
 \begin{align}\label{eq:h}
		h(x)\triangleq
        {\tilde{\mE}\left[\me^{-\frac{1}{2}\boldsymbol{R} x}\right]}\left/{\tilde{\mE}\left[\boldsymbol{R}\me^{-\frac{1}{2}\boldsymbol{R} x}\right]}\right., \quad x\in[0,\infty).
	\end{align}
(Occasionally, we write $h_{\boldsymbol{R}}$ for $h$ to emphasize its dependence on $\boldsymbol{R}$.)
    \end{theorem}
    
\begin{proof}
It suffices to show that for any $t\in[0,T)$, (\ref{limsup:J}) is equivalent to \eqref{eq:equilibrium:g}.

When $\tilde{\mE}[\boldsymbol{R}]<\infty$, we have $h(0) = {1}/{\tilde{\mE}[\boldsymbol{R}]}\in(0,\infty)$ and  $h\in C([0,\infty))$. In this case, an argument similar to the proof of \citet*[Theorem 3.4]{liang2025short} yields the equivalence.

When $\tilde{\mE}[\boldsymbol{R}] = \infty$,  we have $h(0)=0$. In this case,  similar to the case  $\tilde{\mE}[\boldsymbol{R}]<\infty$, \eqref{limsup:J} and \eqref{eq:equilibrium:g} are equivalent for \(t\in [0, T_0)\), where
$T_0 \triangleq \inf\{t\in[0,T]:\ v_a(t) = 0\}$.
Moreover, they are automatically satisfied for $t\in[T_0,T)$ based on Theorem~\ref{thm:0:equi} and $h(0)=0$. 
\end{proof}

Now we introduce a function $l:[0,\infty) \to (0,\infty)$, which is defined by 
$$l(y) \triangleq \tilde{\mE}[e^{-\boldsymbol{R}y}],\quad y\ge0.$$
Its derivative $l'(0)$ exists and is finite for all $y>0$.  In the case $\tilde{\mE}[\boldsymbol{R}]<\infty$, we have $l'(0)=-\tilde{\mE}[\boldsymbol{R}]$. 
In the case $\tilde{\mE}[\boldsymbol{R}]=\infty$, we extend the definition by setting $l'(0)=-\infty$. Then, for any $x\ge0$, the function $h$ is related to $l$ through $h(x) = -{l(x/2)}/{l'(x/2)}$.
The function $h$ defined in \eqref{eq:h} plays a central role in our analysis. Here, we present two examples for which $h$ can be explicitly worked out.

\begin{example} \label{ex:h_calculation}
     Consider the following two distributions for the RRA $\boldsymbol{R}$.
    
    \begin{enumerate}[label=(\arabic*), font=\upshape]
        \item Suppose that  $\boldsymbol{R}$ follows a Poisson distribution with parameter $\theta > 0$. In this case, $l(y) =  e^{\theta(e^{-y}-1)}$ for any $y\ge0$. According to (\ref{eq:h}), we obtain
        \begin{align*}
            h(x) = -{l({x}/{2})}/{l'({x}/{2})} = \frac{1}{\theta} e^{{x}/{2}},\quad x\ge0.
        \end{align*} 
        \item Suppose that $\boldsymbol{R}$ follows a Gamma distribution with shape parameter $\alpha > 0$ and scale parameter $\beta > 0$. Its probability density function is
        \begin{align*}
        f_{\boldsymbol{R}}(r) = \frac{r^{\alpha-1}e^{-{r}/{\beta}}}{\Gamma(\alpha)\beta^\alpha}, \quad r > 0.
        \end{align*}
        In this case, $l(y) = \tilde{\mE}[e^{-\boldsymbol{R}y}] = (1 + \beta y)^{-\alpha}$ ($y\ge0$), which gives
        \begin{align*}
            h(x) = -\frac{l(\frac{x}{2})}{l'(\frac{x}{2})} = \frac{1}{\alpha \beta} \left(1 + \frac{\beta}{2}x\right),\quad x\ge0.
        \end{align*}
    \end{enumerate}
\end{example}

In the two cases of the above example, the functions $h$ are unbounded, violating the boundedness assumption on $h$ made by \cite*{liang2025short}. Moreover, as we will see in 
Example~\ref{exam:poi:gamma} below, the equilibrium strategies can be explicitly obtained for both cases. To accommodate such interesting cases,  
we will conduct a thorough analysis on Condition \eqref{eq:equilibrium:g} and the equilibrium strategies under a general distribution of $\boldsymbol{R}$.

\section{Equilibrium Strategy}\label{Equilibrium Analysis and Characterization of Solutions}

 We first transform Condition \eqref{eq:equilibrium:g} into an ordinary differential equation (ODE) for $v_a$.
 Differentiating $\displaystyle v_a(t) = \int_t^T |a(s)|^2 \md s$ w.r.t. $t$,  we get from (\ref{eq:equilibrium:g}) that $v_a$ satisfies the following ODE:
\begin{align}\label{eq:firstode}
\begin{cases}
v'(t) = -h^2(v(t)) |\lambda(t)|^2, & t \in \left[0, T\right), \\
v(T) = 0.
\end{cases}
\end{align}
Conversely, assume that $v$ solves ODE \eqref{eq:firstode} and let $a(\cdot)\triangleq h(v(\cdot))\lambda(\cdot)$. Then $a$ satisfies \eqref{eq:equilibrium:g}. 

ODE \eqref{eq:firstode} is an equation with separated variables.  We can obtain its solutions in closed form.  

Now we introduce two functions. The first function $\Lambda$ is defined by
\begin{align*}
    \Lambda(t) \triangleq \int_t^T |\lambda(s)|^2 \md s, \quad t\in [0,T].
\end{align*}
The second function $H$ is defined by
\begin{align}\label{eq:H}
\displaystyle H(y) \triangleq \int_{0}^{y} \frac{1}{h^2(x)} \md x=\int_0^y \left({\tilde{\mE}\left[\boldsymbol{R}\me^{-\frac{1}{2}\boldsymbol{R} x}\right]}\left/{\tilde{\mE}\left[\me^{-\frac{1}{2}\boldsymbol{R} x}\right]}\right.\right)^2\md x,\quad y\in[0,\infty].
\end{align}
Occasionally, we write $H_{\boldsymbol{R}}$ for $H$ to emphasize its dependence on $\boldsymbol{R}$.

If $\Lambda(0) = 0$, then $\lambda(\cdot) \equiv \mathbf{0}$ on $[0,T)$ and, in view of \eqref{eq:equilibrium:g}, there is a unique equilibrium $\barpi \equiv \mathbf{0}$. This trivial case is excluded by making the following standing assumption.

\begin{assumption}\label{ass:Psi}  
$\Lambda(0)\in(0,\infty)$. 
\end{assumption}

The next theorem provides,  in closed form,  the solutions to ODE \eqref{eq:firstode} and the equilibrium strategies in $\Pi_d$. We will use $\Pi_d^e$ to denote all equilibrium strategies in $\Pi_d$, that is,
$$\Pi_d^e\triangleq\{\pi\in\Pi_d\,:\ \pi \text{ is an equilibrium strategy}\}.$$

\begin{theorem}\label{theorem3.3}
 Under Assumption \ref{ass:Psi}, we have the following assertions.
\begin{enumerate}[label=(\arabic*), font=\upshape]
    \item If $\mEt[\boldsymbol{R}] < \infty$, then ODE \((\ref{eq:firstode})\) admits a solution if and only if $H(\infty) > \Lambda(0)$. Under this condition, the solution is unique, given by $v(\cdot) = H^{-1}(\Lambda(\cdot))$, and the corresponding unique equilibrium  is  $(\sigma^\top(\cdot))^{-1}h(v(\cdot))\lambda(\cdot)$.

    \item If $\mEt[\boldsymbol{R}] = \infty$ and $H(\varepsilon) < \infty$ for all $\varepsilon > 0$, then ODE \eqref{eq:firstode} admits a one-parameter family of solutions, indexed by $T_0\in\mathcal{T}$:
     \begin{align}\label{solutions}
        v^{(T_0)}(t) \triangleq
        \begin{cases}
            H^{-1}\left(\int_{t}^{T_0} |\lambda(s)|^2 \md s\right), & 0 \le t < T_0, \\
            0, & T_0 \le t \le T,
        \end{cases}
    \end{align}
   where 
    \begin{align*}
    \mathcal{T} \triangleq \left\{ T_0 \in [0, T] : H(\infty) > \Lambda(0) - \Lambda(T_0) \text{ and } \Lambda(t) > \Lambda(T_0)\text{ for all } t \in [0, T_0) \right\}.
\end{align*}
   In this case, 
    \begin{align*}
        \Pi^e_d=\left\{\pi^{(T_0)}(\cdot)\triangleq(\sigma^\top(\cdot))^{-1}h(v^{(T_0)}(\cdot))\lambda(\cdot):\  T_0\in\mathcal{T}\right\}.
    \end{align*}

    \item If $\mEt[\boldsymbol{R}] = \infty$ and $H(\varepsilon) = \infty$ for some $\varepsilon > 0$,  then ODE \((\ref{eq:firstode})\) admits a unique solution $v \equiv 0$ and $\barpi \equiv \mathbf{0}$ is the unique equilibrium in $\Pi_d$.
\end{enumerate}
\end{theorem}

\begin{proof} \ 
\paragraph{(1)} Assume $\mEt[\boldsymbol{R}] < \infty$. In this case, $h(x) > 0$ for all $x \ge 0$.   Suppose that $v(\cdot)$ is a solution to ODE \eqref{eq:firstode}. Separating variables in \eqref{eq:firstode} yields
\begin{align}\label{eq:separate:ode}
 -|\lambda(t)|^2=\frac{v'(t)}{h^2(v(t))}, \quad t\in[0,T).
\end{align}
Integrating \eqref{eq:separate:ode} from $t$ to $T$ yields
\begin{align*}
-\Lambda(t)=\int_t^T \frac{v'(s)}{h^2(v(s))} \md s = \int_{v(t)}^{0} \frac{\md y}{h^2(y)} = - \int_{0}^{v(t)} \frac{\md y}{h^2(y)} = -H(v(t)),
\end{align*}
where the second equality follows from the change of variables $y = v(s)$.  Thus, $H(v(t)) = \Lambda(t)$ for all $t\in[0,T]$. Because \(\displaystyle h(0) = {1}/{\mEt[\boldsymbol{R}]} > 0\), it follows that $H(0)=0$ and $H(\cdot)\in {C}^1\left([0,\infty)\right)$ is strictly increasing. Thus, the inverse function $H^{-1}$  is well-defined, continuous, and strictly increasing in $\left[0, H(\infty)\right)$.
Thus, $H(\infty)>H(v(0))=\Lambda(0)$ and $v(t) = H^{-1}(\Lambda(t))$ for all $t\in[0,T]$, which yields the uniqueness of the solution. 
Conversely, if $H(\infty) > \Lambda(0)$, it is straightforward to verify that $v(\cdot) \triangleq H^{-1}(\Lambda(\cdot))$ indeed solves uniquely the ODE.

\paragraph{(2)} Assume that $\mEt[\boldsymbol{R}] = \infty$ and $H(\varepsilon) < \infty$ for all $\varepsilon > 0$.  We can see that $H(0)=0$, $\displaystyle H\in C^0([0,\infty])\cap C^{1}((0,\infty))$, and $H$ is strictly increasing. Suppose that $v(\cdot)$ is a solution to ODE \eqref{eq:firstode}. Let $T_0\triangleq \inf\{t \in [0, T] :\  v(t) = 0\}$.  
Then $v(t)=0$ for any $t\in [T_0, T]$ and $v(t)>0$ for any $t \in [0, T_0)$.   For any $t\in[0,T_0)$, integrating \eqref{eq:separate:ode} from  $t$ to $T_0$ yields
\begin{align*}
\Lambda(T_0) - \Lambda(t)=\int_t^{T_0} -|\lambda(s)|^2 \md s=\int_{v(t)}^{v(T_0)} \frac{\md y}{h^2(y)} =-H(v(t)).
\end{align*}
Thus, $H(v(t)) = \Lambda(t) - \Lambda(T_0)$ for all $t\in[0,T_0)$.  Therefore, $H(\infty)>H(v(0))=\Lambda(0)-\Lambda(T_0)$ and $\Lambda(t)-\Lambda(T_0)=H(v(t))>H(0)=0$ for all $t\in[0,T_0)$.  Consequently, $T_0\in\mathcal{T}$.
Moreover, for any $t<T_0$, we have 
\begin{align*}
v(t) = H^{-1}(\Lambda(t) - \Lambda(T_0))=H^{-1}\left(\int_{t}^{T_0} |\lambda(s)|^2 \md s\right),
\end{align*}
which implies $v=v^{(T_0)}$.
Conversely, a direct verification shows that \eqref{solutions} indeed solves the ODE.

\paragraph{(3)} Assume that $\mEt[\boldsymbol{R}] = \infty$ and $H(\varepsilon) = \infty$ for some $\varepsilon > 0$. This implies that $\displaystyle \int_0^y \frac{1}{h^2(x)} \md x=\infty$  for any $y > 0$. Clearly, \( v \equiv  0 \) is a solution (recall that \(\displaystyle h(0) = {1}/{\mEt[\boldsymbol{R}]} = 0\)). Suppose, on the contrary, that there exists another solution $v(\cdot)$ with $v(t) > 0$ for some $t < T$. Let $t' \triangleq \inf \{ s \in (t, T] :\  v(s) = 0 \}$. By continuity, we have $v(t')=0$ and $v(s)>0$ for any  $s\in [t, t')$. Integrating ODE \eqref {eq:separate:ode} from $t$ to $t'$ yields
\begin{align*}
\Lambda(t')-\Lambda(t)=\int_t^{t'} -|\lambda(s)|^2 \md s=\int_{v(t)}^{0} \frac{\md y}{h^2(y)} =-\infty,
\end{align*}
which is impossible. Therefore, ODE \eqref{eq:firstode} admits the unique solution  \( v \equiv 0 \).
\end{proof}

Assertion (2) of Theorem~\ref{theorem3.3} provides, in closed form, the solutions to ODE \eqref{eq:firstode} via the parameter $T_0\in\mathcal{T}$. The following lemma provides an alternative explicit characterization of the set $\mathcal{T}$ via an auxiliary mapping.

\begin{lemma}\label{lemma4.4}
Let $\mathcal{A}\triangleq[0,\Lambda(0)]\cap(\Lambda(0)-H(\infty),\Lambda(0)]$ and $$ \quad \varphi(\eta) \triangleq \min \{ t \in [0, T] :\, \Lambda(t) = \eta \},\quad \eta\in[0,\Lambda(0)].$$ 
Then we have
    $\mathcal{T} = \{\varphi(\eta) :\, \eta \in \mathcal{A}\}$. 

\begin{proof} 
From the definition of $\mathcal{T}$, for any $T_0\in\mathcal{T}$, we have $\Lambda(T_0)\in\mathcal{A}$ and
\begin{align*}
T_0 = \inf \{ t \in [0, T] : \Lambda(t) = \Lambda(T_0) \}=\varphi(\Lambda(T_0)).
\end{align*}
Thus, $\mathcal{T}\subset\{\varphi(\eta) :\, \eta \in \mathcal{A}\}$. Conversely, for any $\eta\in\mathcal{A}$,  the continuity of $\Lambda$ implies  $\Lambda(\varphi(\eta))=\eta\in\mathcal{A}$. Moreover, for any $t\in[0,\varphi(\eta))$, we have $\Lambda(t)>\Lambda(\varphi(\eta))$ by the definition of $\varphi$. Therefore, $\varphi(\eta)\in \mathcal{T}$ and hence $\{\varphi(\eta) :\,  \eta \in \mathcal{A}\}\subset \mathcal{T}$.
    \end{proof}
\end{lemma}

\begin{remark}\label{rmk:infinity}
The sets $\mathcal{T}$ and $\mathcal{A}$ are in one-to-one correspondence because $\varphi$ is strictly decreasing. It is easily seen that $ \mathcal{A} $ contains infinitely many elements. Consequently, under the conditions $\mEt[\boldsymbol{R}] = \infty$ and $H(\varepsilon) < \infty$ for all $\varepsilon > 0$, we obtain infinitely many equilibria.
\end{remark}

It should be noted that the existing studies, such as \cite{Desmettre2023} and \cite*{liang2025short}, focus mainly on the specific cases covered by Theorem~\ref{theorem3.3}(1), under the additional assumptions that $h(\cdot)$ is bounded and $\boldsymbol{R}$ is square-integrable. Our result covers any general RRA $\boldsymbol{R}$. In the following two examples, the functions $h$ are unbounded and the equilibrium strategies are explicitly worked out.

\begin{example}\label{exam:poi:gamma}
This example derives the explicit equilibrium strategies for the two distributions discussed in Example \ref{ex:h_calculation}.
\begin{enumerate}[label=(\arabic*), font=\upshape, wide, labelwidth=!, labelindent=0pt]
    \item \textbf{Poisson distribution:} Recall that $h(x) = \dfrac{1}{\theta} e^{\frac{x}{2}}$. Then $h$ is unbounded and
    \begin{align*}
        H(y) = \int_0^y \frac{1}{h^2(x)} dx = \theta^2 (1 - e^{-y}), \quad y\in[0,\infty),
    \end{align*}
    with $H(\infty) = \theta^2$. If $\Lambda(0) < \theta^2$, then $v(t) = H^{-1}(\Lambda(t))= -\ln\left(1 - \dfrac{\Lambda(t)}{\theta^2}\right)$. In this case, the unique equilibrium strategy is explicitly given by
    \begin{align*}
        \barpi(t)=(\sigma^\top(t))^{-1} a(t)= (\sigma^\top(t))^{-1}h(v(t))\lambda(t) = \dfrac{(\sigma^\top(t))^{-1}\lambda(t)}{\sqrt{\theta^2 - \Lambda(t)}},\quad t\in[0,T).
    \end{align*}

    \item \textbf{Gamma distribution:} Recall that $h(x) = \dfrac{1}{\alpha \beta} (1 + \dfrac{\beta}{2}x)$. Then $h$ is unbounded and
    \begin{align*}
        H(y) = \int_0^y \frac{\alpha^2 \beta^2}{(1 + \frac{\beta}{2}x)^2} dx = 2\alpha^2 \beta \left(1 - \dfrac{1}{1 + \frac{\beta}{2}y}\right),\quad y\in[0,\infty),
    \end{align*}
    with $H(\infty) = 2\alpha^2 \beta$. Provided that $\Lambda(0) < 2\alpha^2 \beta$, we obtain $v(t) = \dfrac{2}{\beta} \left[ \left( 1 - \dfrac{\Lambda(t)}{2\alpha^2\beta} \right)^{-1} - 1 \right]$. In this case, the unique equilibrium strategy is explicitly given by
    \begin{align*}
        \barpi(t)=(\sigma^\top(t))^{-1} a(t)= \dfrac{(\sigma^\top(t))^{-1}\lambda(t)}{\alpha \beta - \frac{\Lambda(t)}{2\alpha}},\quad t\in[0,T).
    \end{align*}
\end{enumerate}
\end{example}

\begin{example} 
This example is provided for Assertions (2) and (3) in Theorem \ref{theorem3.3}.  Let $\boldsymbol{R}_{\alpha}\sim G_\alpha$, where  $\alpha\in(0,1)$ is a fixed parameter and $G_\alpha$ denotes the  distribution function whose Laplace transform satisfies (see \cite{Pollard1946} and \citet[Section XIII.6]{Feller1971})
$$\int_0^\infty \me^{-vx}\md G_\alpha(x)=\me^{-v^\alpha},\quad v\ge0.$$
    Note that $\mEt[\boldsymbol{R}_{\alpha}] = \infty$ for $\alpha \in (0, 1)$. Moreover, $\displaystyle h(x) = \frac{1}{\alpha}\left(\frac{x}{2}\right)^{1-\alpha}$, $x\ge0$. Obviously, $h$ is unbounded. For simplicity,  we assume that $d=1$, $\sigma(\cdot)\equiv\sigma\neq0$ and \(\lambda(\cdot)\equiv\lambda \neq0\) are constants in \(\mR\).  
    \begin{itemize}
        \item Let $\alpha \in (0.5,1)$. By direct calculation, the solutions to ODE~\eqref{eq:firstode}  take the form
    \begin{align*}
        v^{(T_0)}(t)=
        \begin{cases}
            \left(2^{2\alpha-2}\dfrac{1}{\alpha^2}(2\alpha-1)\lambda^2(T_0-t)\right)^{\frac{1}{2\alpha-1}},&t\in[0,T_0),\\
            0, & t\in[T_0,T],
        \end{cases}
    \end{align*}
    where \(T_0\in[0,T]\). Consequently, the explicit expression for the equilibrium strategy $\pi^{(T_0)}(t)$ on $[0, T)$ is
    \begin{align*} 
    \pi^{(T_0)}(t)&= 
    \begin{cases} (\sigma)^{-1} \frac{1}{\alpha} \left[ \frac{1}{2} \left( 2^{2\alpha-2} \frac{2\alpha-1}{\alpha^2} \lambda^2 (T_0 - t) \right)^{\frac{1}{2\alpha-1}} \right]^{1-\alpha}\lambda, & 0 \le t < T_0, \\
    0, & T_0 \le t < T， 
    \end{cases} \\
    &=\begin{cases} (\sigma)^{-1} \left( \frac{\lambda}{\alpha} \right)^{\frac{1}{2\alpha-1}} \left[ \frac{2\alpha-1}{2}(T_0 - t) \right]^{\frac{1-\alpha}{2\alpha-1}}, & 0 \le t < T_0, \\ 0, & T_0 \le t < T,
    \end{cases}
    \end{align*}
    where \(T_0\in[0,T]\). 
    
    \item Let $\alpha\in(0,0.5]$. Direct calculation shows that  ODE \eqref{eq:firstode} admits only the trivial solution $v \equiv 0$ and therefore $\barpi=0$ is the unique equilibrium.
\end{itemize}
\end{example}

Assume $\mEt[\boldsymbol{R}] < \infty$ and $H(\infty)>\Lambda(0)$. Assertion (1) of Theorem~\ref{theorem3.3} shows that ODE \eqref{eq:firstode} admits a unique solution $v$ and  therefore the deterministic equilibrium strategy exists and is unique. To emphasize the dependence on $\boldsymbol{R}$, we occasionally write $v_{\boldsymbol{R}}$ and $a_{\boldsymbol{R}}$ for $v$ and $a$, respectively.

\begin{remark}
The closed form of equilibrium strategy in Theorem \ref{theorem3.3}(1) yields a clear structural decomposition of the risk exposure vector $a$, given by
\begin{align}\label{eq:localvola}
    a(t) = h(H^{-1}(\Lambda(t)))\lambda(t), \quad t\in[0,T).
\end{align}
Note that the functions $h(\cdot)$ and $H(\cdot)$ depend solely on the distribution of the RRA ${\boldsymbol{R}}$. The representation \eqref{eq:localvola} shows that, at time $t$, the effects of three distinct factors on the risk exposure vector are totally separated: the instantaneous market opportunity $\lambda(t)$, the aggregate future market opportunities $\Lambda(t)$ over the remaining interval $(t,T)$, and the RRA itself (represented by $h$ and $H$).

More specifically, at time $t$, if $\Lambda(t)$ and ${\boldsymbol{R}}$ are fixed, a larger current risk premium $|\lambda(t)|$ leads to a larger $|a(t)|$. If $|\lambda(t)|$ and ${\boldsymbol{R}}$ are fixed, then a larger future market opportunity $\Lambda(t)$ implies a larger value of $H^{-1}(\Lambda(t))$, and hence a larger value of $h\big(H^{-1}(\Lambda(t))\big)$ (see Proposition \ref{proposition3.10} for  a detailed discussion; see also Lemma \ref{h:mono} for the monotonicity of $h$). Finally, if both $\lambda(t)$ and $\Lambda(t)$ are fixed, one can study how ${\boldsymbol{R}}$ affects $h$ and $H$, and thereby its impact on $a(t)$; see the comparative statics analysis in Section~\ref{Comparative Statics of Equilibrium Strategies}.

We emphasize that, in contrast to the classical CRRA utility, which generates a “myopic” strategy reacting primarily to the instantaneous market price of risk $\lambda(t)$, the equilibrium under RRA is inherently forward-looking due to the dependence on $\Lambda(t)$. Finally, we refer to \cite{liang2024dynamic} and \cite*{liang2025dynamic} for similar structures of the equilibrium strategies under other preferences, although monotonicity properties and comparative statics are not investigated therein.
\end{remark}

Next, we study the effect of changes in future market opportunities on the risk exposure vector.
\begin{proposition}\label{proposition3.10}
Let $\boldsymbol{R}$ be a non-negative random variable with $\mEt[\boldsymbol{R}]\in(0,\infty)$. Fix a vector $\lambda_0 \in \mathbb{R}^d$. Consider a sequence of market settings indexed by $n \in \mathbb{N}$, defined by the pairs $\{(\lambda_n(\cdot), T_n)\}_{n \ge 1}$, where $T_n > 0$ and $\lambda_n: [0, T_n] \to \mathbb{R}^d$ satisfy $\lambda_n(0) = \lambda_0$. Let $\displaystyle\Lambda_n \triangleq \int_{0}^{T_n} |\lambda_n(s)|^2 ds$. Let $a_n(0)$ denote the risk exposure vector at time $0$ corresponding to the $n$-th market setting.
\begin{enumerate}[label=(\arabic*), font=\upshape]
    \item If the sequence $\{\Lambda_n\}_{n \ge 1}$ strictly increases to $H(\infty)$, then
   $\lim\limits_{n \to \infty} a_n(0) = \frac{\lambda_0}{r_0}$,
    where $r_0 \triangleq \operatorname{essinf} \boldsymbol{R}$. Moreover,   $|a_n(0)| \nearrow \frac{|\lambda_0|}{r_0}$ as $n\to \infty$.
    
    \item If the sequence $\{\Lambda_n\}_{n \ge 1}$ strictly decreases to  $0$, then
$\lim\limits_{n \to \infty} a_n(0) = \frac{\lambda_0}{\mEt[\boldsymbol{R}]}$.
    Moreover, 
    $|a_n(0)| \searrow \frac{|\lambda_0|}{\mEt[\boldsymbol{R}]}$ as $n\to \infty$.
\end{enumerate}
Here, we adopt the conventions $\frac{1}{0}=\infty$ and $0\cdot\infty=0$.
\begin{proof}
See Appendix \ref{proof of proposition3.10}
\end{proof}
\end{proposition}

\begin{remark}
The convergence results in Proposition \ref{proposition3.10} depend solely on the sequence of integral values $\{\Lambda_n\}_{n \ge 1}$, and no convergence of the parameter sequence $\{(\lambda_n(\cdot), T_n)\}_{n \ge 1}$ itself is required. Moreover, focusing on time $0$ does not cause a loss of generality. For any time $t$, analogous conclusions follow by a simple time translation, with $t$ treated as the initial time.
\end{remark}
 \begin{remark}
    In the single-stock case, assuming that $\lambda > 0$ and $\sigma > 0$ are constants, one can use $\widehat{\pi}(\cdot) = \frac{1}{\sigma} h(v(\cdot)) \lambda$, together with the monotonicity of $h$ given in  Lemma \ref{h:mono}, to establish a favorable property: $\pi(\cdot)$ is decreasing. See Section 4.3 of \cite{Desmettre2023} for a detailed discussion.    However, the results in \cite{Desmettre2023} do not yield Proposition~\ref{proposition3.10}. The proof of Proposition~\ref{proposition3.10} is crucially based on our closed-form expression for the risk exposure vector, given in \eqref{eq:localvola}, which is not available in their work.
\end{remark}

As the final part of this section, we use the closed form of the equilibria in Theorem \ref{theorem3.3} to investigate the continuity properties of $a(\cdot)$ w.r.t. the distribution of the random risk aversion. More specifically,
the following theorem shows that if a sequence of RRA $\{\boldsymbol{R}_n\}_{n\ge1}$ converges in distribution to $\boldsymbol{R}$, then the corresponding sequence of risk exposure vectors $\{a_n\}_{n\ge1}$  converges to $a$.

\begin{theorem}\label{thm:converge:distribution}
Let $\{\boldsymbol{R}_n\}_{n\ge 1}$ be a sequence of random risk aversions that converges in distribution to $\boldsymbol{R}$ and $\lim\limits_{n \to \infty} \mEt[\boldsymbol{R}_n] = \mEt[\boldsymbol{R}] \in (0, \infty)$. Assume that 
$H_n(\infty) > \Lambda(0)$ for all $n \geq 1$ and 
$H(\infty) > \Lambda(0)$.
Let $a_n=a_{\boldsymbol{R}_n}$ and $a=a_{\boldsymbol{R}}$.
Then $\lim\limits_{n \to \infty} a_n(t) = a(t)$ for all $t \in [0, T)$. Moreover, if $\lambda \in L^\infty([0,T))$, then the convergence is uniform on $[0,T)$.
\begin{proof}
See Appendix \ref{proof:thm:converge:distribution}.
\end{proof}
\end{theorem}

\section{Optimal Equilibrium}\label{Optimal Equilibrium Selection}

Remark \ref{rmk:infinity} shows that, if $\mEt[\boldsymbol{R}] = \infty$ and $H(\varepsilon) < \infty$ for all $\varepsilon > 0$, then the set $\Pi^e_d$ of the equilibria contains infinitely many elements. 
This multiplicity of equilibria raises the natural question: which element of $\Pi^e_d$ should the agent choose?
In this section, we propose that the agent  choose the equilibrium $\pi$ that maximizes objective functional $J(\pi;0,w)$ at the initial time over all $\pi\in\Pi^e_d$. This leads to the following definition.

\begin{definition} An equilibrium
\( \widehat \pi \in\Pi^e_d \) is called optimal if
\(J(\widehat \pi;0,w) \ge J( \pi;0,w)\) for all \(\pi\in\Pi^e_d 
\) and for all $w>0$.
\end{definition}

Recalling the relationship $J(\pi;t,w)=wJ_0(t,\pi)$ in Remark \ref{rmk:j_0}, finding an optimal equilibrium in $\Pi^e_d$ is equivalent to finding a $\widehat{\pi}\in\Pi^e_d$ that maximizes $J_0(t,\pi)$. 
 
 Let $\pi\in\Pi^e_d$.  
 We begin by computing $J_0(t,\pi)$.  Using Theorem~\ref{theorem3.3}(2), we have $\pi=\pi^{(T_0)}$ for some $T_0\in\mathcal{T}$. Let $a^{(T_0)}\triangleq \sigma^{\top}\pi^{(T_0)}$ and $
    y^{(T_0)}(t) \triangleq \int\limits_{t}^{T} (a_s^{(T_0)})^\top \lambda(s) \md s $. Then
\begin{align*}
-(\boldsymbol{R} v^{(T_0)}(t) - 2 y^{(T_0)}(t)) &=-\boldsymbol{R} v^{(T_0)}(t) + 2\int_{t}^{T} (a_s^{(T_0)})^\top \lambda(s) \md s\\
&=-\boldsymbol{R} v^{(T_0)}(t) + 2 \int_{t}^{T} h(v^{(T_0)}(s)) |\lambda(s)|^2 \md s \\&= -\boldsymbol{R} v^{(T_0)}(t) + 2 \int_{0}^{v^{(T_0)}(t)} \frac{1}{h(y)} \md y,
\end{align*}
where the third equality arises from the change of variable \( y = v^{(T_0)}(s) \).
Thus, we obtain 
\begin{align*}
J_0(t, \pi^{(T_0)}) = \exp\left(\int_{0}^{v^{(T_0)}(t)} \frac{1}{h(y)} \md y\right) l\left( \frac{v^{(T_0)}(t)}{2} \right).
\end{align*}
Let $\mathcal{L}(z):[0,\infty)\to \mR$ be defined by
\(\mathcal{L}(z) =\int_{0}^{z} \frac{1}{h(y)} \md y + \log l\left( \frac{z}{2} \right).\)
Then we have 
\begin{align}\label{form:J}
    J_0(t, \pi^{(T_0)}) = \exp\left( \mathcal{L}(v^{(T_0)}(t)) \right).
\end{align}

According to  \eqref{form:J}, in order to study the
maximization of $J_0 (t, \pi^{(T_0)})$ over $T_0 \in \mathcal{T}$, we need to analyze
the monotonicity of $\mathcal{L}$ and the monotonicity of $v^{(T_0)}$ w.r.t. 
$T_0$. These properties are provided in Lemma 
\ref{lemma4.2} and Lemma \ref{lemma4.3}, respectively.
\begin{lemma}\label{lemma4.2}
    $\mathcal{L}(z)$ is strictly increasing in $z$.
\begin{proof}
By direct differentiation, we have
    \begin{align*}
\frac{d\mathcal{L}}{dz} (z)= \frac{1}{h(z)} - \frac{1}{2} \frac{\mEt \left[ \boldsymbol{R} e^{-\frac{1}{2} \boldsymbol{R} z} \right]}{\mEt \left[ e^{-\frac{1}{2} \boldsymbol{R} z} \right]}=\frac{1}{2h(z)}>0 \quad \text{for all } z>0.
\end{align*}
As such, $\mathcal{L}(z)$ is strictly increasing in $z$.
\end{proof}
\end{lemma}

\begin{lemma}\label{lemma4.3}
Under the conditions of Assertion (2) of 
Theorem \ref{theorem3.3}, 
let $T_{0,1}, T_{0,2} \in \mathcal{T}$ with $T_{0,1} < T_{0,2}$. Then 
\begin{align}\label{ineq:vT1T2}
v^{(T_{0,1})}(t)
\begin{cases}
< v^{(T_{0,2})}(t),\quad t\in[0,T_{0,2}),\\
= v^{(T_{0,2})}(t),\quad t\in[T_{0,2},T].
\end{cases}
\end{align}
In particular, if $H(\infty) > \Lambda(0)$ and $\Lambda(t) > 0$ for all $t \in [0, T)$, then $T\in\mathcal{T}$ and $v^{(T)}(t)>  v^{(T_0)}(t)$ for any $T_0 \in \mathcal{T} \setminus \{T\}$ and any $t \in [0, T)$.
\end{lemma}

\begin{proof}
    Let $T_{0,1}, T_{0,2} \in \mathcal{T}$ with $T_{0,1} < T_{0,2}$. 
For $t\in[T_{0,2},T]$, we have $v^{(T_{0,1})}(t) = v^{(T_{0,2})}(t) = 0$ based on  the definition of $ \mathcal{T} $. 
For $t\in[T_{0,1}, T_{0,2})$, we have $v^{(T_{0,2})}(t) > 0=v^{(T_{0,1})}(t)$.
For $t\in[0, T_{0,1})$, we have
\begin{align*}
v^{(T_{0,i})}(t) = H^{-1}\left(\int_t^{T_{0,i}} |\lambda(s)|^2 \md s\right) = H^{-1}(\Lambda(t) - \Lambda(T_{0,i})), \quad i=1,2.
\end{align*}
By the definition of $\mathcal{T}$, $\Lambda(T_{0,1}) > \Lambda(T_{0,2})$. 
Because $H^{-1}$ is strictly increasing, we have $v^{(T_{0,1})}(t) < v^{(T_{0,2})}(t)$ for $t\in[0, T_{0,1})$. Thus, \eqref{ineq:vT1T2} is proved.

Assume $H(\infty) > \Lambda(0)$ and $\Lambda(t) > 0$ for all $t \in [0, T)$. 
By the definition of $\mathcal{T}$, $T\in\mathcal{T}$. 
Plugging $T_{0,2}=T$ and $T_{0,1}=T_0$ into \eqref{ineq:vT1T2} yields $v^{(T)}(t)>  v^{(T_0)}(t)$ for any $T_0 \in \mathcal{T} \setminus \{T\}$ and any $t \in [0, T)$.
\end{proof}

We are now ready to present the main result on the optimal equilibria.

\begin{theorem}\label{thm:optimal}
Under the conditions of Assertion (2) of Theorem \ref{theorem3.3}, we have the following assertions.
\begin{enumerate}[label=(\arabic*), font=\upshape, leftmargin=*]
\item The optimal strategy exists if and only if  \( H(\infty) > \Lambda(0) \). If it exists, it is unique and given by \( \pi^{(\varphi(0))}(\cdot)=(\sigma^\top(\cdot))^{-1}h( v^{(\varphi(0))}(\cdot))\lambda(\cdot)\).

 \item $\pi^{(\varphi(0))}$ is also uniformly optimal:
 \begin{align*}
J( \pi^{(\varphi(0))};t,w) \ge J( \pi;t,w) \text{ for all }t\in[0,T), w>0,  \text{ and } \pi\in\Pi^e_d. 
\end{align*}
 
\item \(\pi^{(\varphi(0))}\) is uniformly strictly optimal, that is,
 \begin{align*}
J( \pi^{(\varphi(0))};t,w) > J( \pi;t,w) \text{ for all }t\in[0,T), w>0,  \text{ and } \pi\in\Pi^e_d\setminus\{\pi^{(\varphi(0))}\},
\end{align*}
if and only if 
\[H(\infty) > \Lambda(0) \text{ and } \Lambda(t) > 0\text{ for all } t \in [0, T).\] If it is the case, then $\varphi(0) = T$.
\end{enumerate}
\begin{proof}\
\paragraph{(1)}  
Combining \eqref{form:J} with Lemma \ref{lemma4.2}, we see that $J_0(0, \pi^{(T_0)})$ attains its maximum precisely when  $v^{(T_0)}(0)$ is maximal. Moreover, by Lemma \ref{lemma4.3},  $\displaystyle v^{(T_0)}(0)$ is strictly increasing in $T_0\in\mathcal{T}$. Thus, finding the optimal equilibrium strategy is equivalent to identifying the maximal parameter  $T_0\in \mathcal{T}$.  By Lemma \ref{lemma4.4}, $\displaystyle\mathcal{T} = \{\varphi(\eta) :\,  \eta \in \mathcal{A}\}$. The existence of a maximal element in $\mathcal{T}$ is equivalent to the existence of a minimal element in $\mathcal{A}$, because $\varphi$ is strictly decreasing on $\mathcal{A}$.

If $H(\infty) > \Lambda(0)$, then $\mathcal{A} = [0, \Lambda(0)]$ includes its minimum at $0$, and hence a optimal strategy exists. Conversely, if $H(\infty) \le \Lambda(0)$, then $\mathcal{A} = (\Lambda(0) - H(\infty), \Lambda(0)]$ is left-open and admits no minimum. Thus, $\mathcal{T}$ has a maximal element if and only if $H(\infty) > \Lambda(0)$.  
In conclusion, the optimal strategy exists if and only if $H(\infty) > \Lambda(0)$, and when it exists, it is given by $\pi^{(\varphi(0))}$.
\paragraph{(2)}
Following the same line of reasoning as in Assertion (1), we immediately obtain that $\pi^{(\varphi(0))}$
 is also uniformly optimal as \eqref{form:J} remains valid for every $t\in[0,T]$.

\paragraph{(3)} 
Assume that $\Lambda(t) = 0$ for some $t \in [0, T)$. Then $v^{(T_0)}(t) \equiv 0$ for any $T_0\in\mathcal{T}$, which implies $J(t, \pi^{(T_0)}) \equiv \exp (\mathcal{L}(0))= 1$ for all $\pi^{(T_0)} \in \Pi_d^e$. Hence, there is no uniformly strictly optimal equilibrium in this case.  Conversely, assume that $\Lambda(t) > 0$ for all $t \in [0, T)$ and $H(\infty) > \Lambda(0)$.  We have $\varphi(0) = T$ because $\Lambda(t) > 0$ for all $t < T$. Let $T_0\in\mathcal{T}\setminus \{T\}$. By Lemma \ref{lemma4.3}, for any $t\in[0,T)$, it holds that 
\begin{align*}
    v^{(T)}(t) > v^{(T_0)}(t).
\end{align*}
We obtain $J_0(t, \pi^{(T)}) > J_0(t, \pi^{(T_0)})$ as $\mathcal{L}(z)$ is strictly increasing in $z$. Therefore, $\pi^{(\varphi(0))}=\pi^{(T)}$ is uniformly strictly optimal.
\end{proof}
\end{theorem}

In analogy with Proposition \ref{proposition3.10}, we now analyze the effect of future market opportunities on the risk exposure vector for the case $\mEt[\boldsymbol{R}]=\infty$. Given the multiplicity of equilibria in this setting, we focus on the optimal equilibrium characterized in Theorem \ref{thm:optimal}. The following proposition presents the corresponding results. We omit the proof as it follows a similar line of reasoning to that of Proposition \ref{proposition3.10}.

\begin{proposition}
    Let $\boldsymbol{R}$ be a non-negative random variable with $\mEt[\boldsymbol{R}] = \infty$. Fix a vector $\lambda_0 \in \mathbb{R}^d$. Consider a sequence of market settings indexed by $n \in \mathbb{N}$, defined by the pairs $\{(\lambda_n(\cdot), T_n)\}_{n \ge 1}$, where $T_n > 0$ and $\lambda_n: [0, T_n] \to \mathbb{R}^d$ satisfy $\lambda_n(0) = \lambda_0$. Let $\Lambda_n \triangleq \int_{0}^{T_n} |\lambda_n(s)|^2 ds$ and assume  $H(\infty) > \Lambda_n$ for all $n$. Let $a_n(0)$ denote the risk exposure vector at time $0$ corresponding to the optimal equilibrium of the $n$-th market setting.
\begin{enumerate}[label=(\arabic*), font=\upshape]
    \item If the sequence $\{\Lambda_n\}_{n \ge 1}$ is strictly increasing and $\Lambda_n \nearrow H(\infty)$, then
    $\lim\limits_{n \to \infty} a_n(0) = \frac{\lambda_0}{r_0}$,
    where $r_0 \triangleq \operatorname{essinf} \boldsymbol{R}$. Moreover, 
    $|a_n(0)| \nearrow \frac{|\lambda_0|}{r_0}$ as $n\to \infty$.
    
    \item If the sequence $\{\Lambda_n\}_{n \ge 1}$ is strictly decreasing and $\Lambda_n \searrow 0$, then
 $\lim\limits_{n \to \infty} a_n(0) = \mathbf{0}$.
    Moreover,      $|a_n(0)| \searrow 0$ as $n\to \infty$.
\end{enumerate}
Here we adopt the convention $\frac{1}{0}=\infty$.
\end{proposition}

\section{Comparative Statics}\label{Comparative Statics of Equilibrium Strategies}

In the previous sections, we have derived, in closed form, the equilibrium strategies associated with a given distribution of the RRA $\boldsymbol{R}$ and discussed how to select the optimal equilibrium. We now turn to comparative statics and consider a natural question: how does a shift in the distribution of the RRA affect the resulting equilibrium investment behavior?  

Throughout this section, we write $H_i$ for $H_{{\boldsymbol{R}_i}}$ and $a_i$ for $a_{{\boldsymbol{R}_i}}$, $i=1,2$.

\subsection{Two-Point Distributed RRA and Single-Crossing}
We begin by considering the first-order stochastic dominance, which should be a natural way to rank the risk aversion distributions.

\begin{definition}
Let ${\boldsymbol{R}_1}$ and ${\boldsymbol{R}_2}$ be two random variables. We say that ${\boldsymbol{R}_1}$ dominates ${\boldsymbol{R}_2}$ in the sense of first-order stochastic dominance, denoted ${\boldsymbol{R}_1} \succeq_{1} {\boldsymbol{R}_2}$, if
$\tilde{\mathbb{P}}({\boldsymbol{R}_1} \ge x) \ge \tilde{\mathbb{P}}({\boldsymbol{R}_2} \ge x)$ for all $x \in \mathbb{R}$.
\end{definition}
Intuitively, one might expect that if  ${\boldsymbol{R}_1} \succeq_{1} {\boldsymbol{R}_2}$, then the corresponding risk exposure magnitudes $|a_i(\cdot)|$ ($i=1,2$) should satisfy $|a_1(t)| \le |a_2(t)|$ for any $t\in[0,T)$. Surprisingly, this intuition turns out to be false, as  the next example with two-point distributed RRA shows.

 \begin{example} \label{exm:fod}
 Assume that there is only one stock, whose market price of risk is constant, $\lambda(\cdot)\equiv\lambda = 0.4$. The time horizon $T = 20$. 
There are two investors with distributions of  RRA:
\begin{itemize}
\item Investor 1: $\tilde{\mathbb{P}}(\boldsymbol{R}_1 = 1)=0.9$ and $\tilde{\mathbb{P}}(\boldsymbol{R}_1 = 3) = 0.1$;
    
\item Investor 2: $\tilde{\mathbb{P}}(\boldsymbol{R}_2 = 1)=0.9$ and $\tilde{\mathbb{P}}(\boldsymbol{R}_2 = 2) = 0.1$. 
\end{itemize} 
Obviously, $\boldsymbol{R}_1\succeq_{1} \boldsymbol{R}_2$.
Figure \ref{counterexample_plot.png} displays the comparative evolution of $|a_i(\cdot)|$ ($i=1,2$).\footnote{Because $\barpi = (\sigma^\top)^{-1} a$, in the one-dimensional case comparing $|a_i(\cdot)|$ is equivalent to comparing $|\pi_i(\cdot)|$.  More generally, $|a|$ represents the local volatility of the investor's wealth process and thus captures the level of risk borne by the investor. Therefore, in what follows, we focus conceptually on comparing the risk exposure magnitude (local volatility) instead of the trading strategy.} 
The red-shaded Reversal Region highlights that Investor 1 adopts a strictly riskier position than Investor 2 when the time to maturity is sufficiently long. This counterexample shows that a ``larger” RRA in the sense of first-order stochastic dominance does not necessarily yield a less risky investment.
\end{example}

\begin{figure}[ht]
    \centering
    \includegraphics[width=0.7\textwidth]{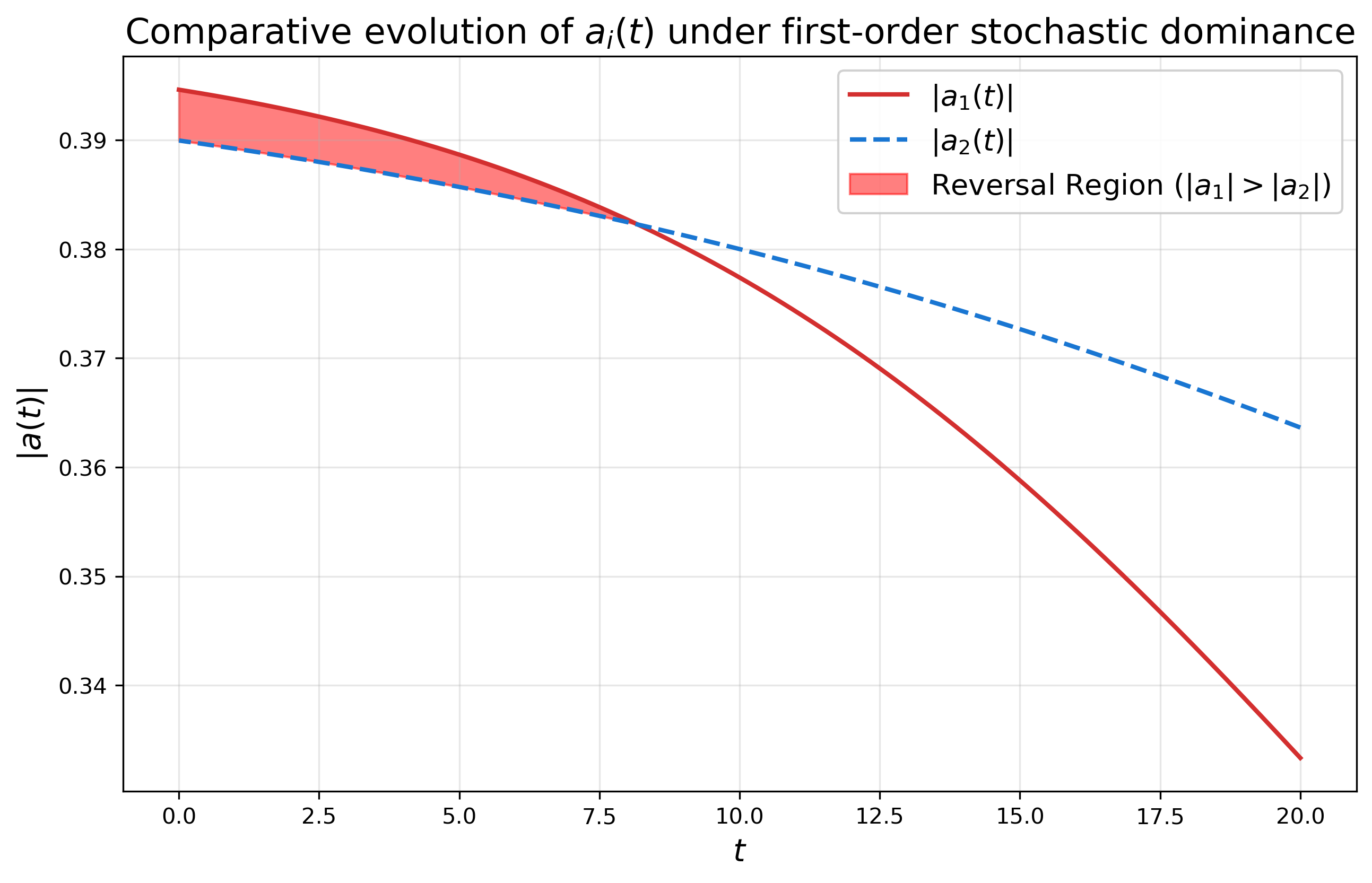}
    \caption{\small Comparative evolution of $|a_i(\cdot)|$ ($i=1,2$) under the first-order stochastic dominance. The red solid line corresponds to Investor 1  with  $\tilde{\mathbb{P}}(\boldsymbol{R}_1=1)=0.9$ and $\tilde{\mathbb{P}}(\boldsymbol{R}_1=3)=0.1$; the blue dashed line corresponds to Investor 2  with  $\tilde{\mathbb{P}}(\boldsymbol{R}_2=1)=0.9$ and $\tilde{\mathbb{P}}(\boldsymbol{R}_2=2)=0.1$. The red shaded region indicates where $|a_1(\cdot)| > |a_2(\cdot)|$. Parameters: $\lambda = 0.4$,  $T = 20$.}
    \label{counterexample_plot.png}
\end{figure}

The appearance of Example \ref{exm:fod} might suggest that establishing a general comparative result under the first-order stochastic dominance is fundamentally obstructed. In fact, the observed reversal is not merely an artifact of a particular parameter choice. 

In the following proposition, we focus on the case where the risk aversions $\boldsymbol{R}_i$ ($i=1,2$) satisfy $\essinf \boldsymbol{R}_1 =\essinf \boldsymbol{R}_2= r_0 > 0$.
In this case, $H_i(\infty)=\infty$ ($i=1,2$) (see Lemma \ref{lemmaA.3}). Consequently, according to Theorem \ref{theorem3.3}(1), for any deterministic and right-continuous market price of risk $\lambda(\cdot)$ in $L^2([0,T))$, the equilibrium strategy exists and is unique.

The next proposition shows that, under suitable conditions on $\boldsymbol{R}_i$, in particular when 
$\boldsymbol{R}_1$ and $\boldsymbol{R}_2$ share the same ``smallest'' value but $\boldsymbol{R}_1$ has a 
larger ``second-smallest'' value than $\boldsymbol{R}_2$, there exists $\overline{\Lambda}>0$ such that, for any 
$\lambda(\cdot)$ with $\lambda(0)\neq 0$  and 
$\int_0^T |\lambda(t)|^2 \, \md t > \overline{\Lambda}$, the corresponding risk exposure magnitudes satisfy 
$|a_1(0)|>|a_2(0)|$. Moreover, if $\mEt[\boldsymbol{R}_1] > \mEt[\boldsymbol{R}_2]$ and $\lambda(T-)\neq 0$, then  $|a_1(T-)| < |a_2(T-)|$. 

\begin{proposition}\label{proposition5.2}
Let $\boldsymbol{R}_1$ and $\boldsymbol{R}_2$ be two non-negative random variables with $\mEt[\boldsymbol{R}_i] < \infty$ $(i=1, 2)$ satisfying the following two conditions.
\begin{enumerate}[label=(\arabic*), font=\upshape]
    \item[(c1)] Let $r_i \triangleq\essinf \boldsymbol{R}_i$, $i=1,2$. Assume that $r_0\triangleq r_1=r_2>0$ and $p_i \triangleq \tilde{\mathbb{P}}(\boldsymbol{R}_i = r_0) \in (0, 1)$ for $i = 1, 2$.
    \item[(c2)] Let $\tilde{r}_i \triangleq \inf\{r \in \supp(\boldsymbol{R}_i) : r > r_0\}$ and $\delta_i \triangleq \tilde{r}_i - r_0$, $i=1,2$. Assume that $\delta_1 > \delta_2>0$.
\end{enumerate}
Then, there exists $\overline{\Lambda}>0$ such that, for any 
$\lambda(\cdot)$ with $\lambda(0)\neq 0$ and 
$\int_0^T |\lambda(t)|^2 \, \md t > \overline{\Lambda}$, the corresponding risk exposure magnitudes satisfy 
$|a_1(0)|>|a_2(0)|$. In addition, if $\mEt[\boldsymbol{R}_1] > \mEt[\boldsymbol{R}_2]$ and $\lambda(T-)\neq 0$, then  $|a_1(T-)| < |a_2(T-)|$.
\begin{proof}
See Appendix \ref{proof of proposition5.2}.
\end{proof}
\end{proposition}

Next, under the conditions of Proposition~\ref{proposition5.2}, we further investigate the crossing properties of the risk exposure magnitudes $|a_i(\cdot)|$ ($i=1,2$) corresponding to the two distributions. This includes the single-crossing property (Proposition~\ref{proposition5.3}) as well as the monotonicity of the crossing time w.r.t. the distribution parameters (Propositions~\ref{proposition5.5} and \ref{proposition5.6}). To this end, we focus on the two-point distributed RRA, which, as investigated by \cite{Desmettre2023}, is  the closest RRA model to the standard expected utility framework. Specifically, we assume that the risk aversion $\boldsymbol{R}_i$ ($i=1,2$) satisfies
\begin{align} \label{eq:star}
\begin{split}
&\tilde{\mathbb{P}}(\boldsymbol{R}_i=r_0)=p_i=1- \tilde{\mathbb{P}}(\boldsymbol{R}_i=r_0+\delta_i) \quad\text{with } r_0 > 0, \, p_i \in (0,1), \delta_1 > \delta_2 > 0.
\end{split}
\end{align}
Below, we begin by establishing the single-crossing property of the risk exposure magnitudes; that is, under certain conditions, there exists a unique time $t^*$ such that $|a_1(t^*)| = |a_2(t^*)|$.
\begin{proposition}\label{proposition5.3}
Suppose that the risk aversions $\boldsymbol{R}_1$ and $\boldsymbol{R}_2$ satisfy \eqref{eq:star} and $\mEt[\boldsymbol{R}_1] > \mEt[\boldsymbol{R}_2]$.\footnote{The condition $\mEt[\boldsymbol{R}_1] > \mEt[\boldsymbol{R}_2]$ is equivalent to $(1-p_1)\delta_1 > (1-p_2)\delta_2$.}
 Let $\widehat{\Lambda} \in (0, \infty)$ be the unique solution to the equation
\begin{align*}
    h_1\left(H_1^{-1}(\widehat{\Lambda})\right) = h_2\left(H_2^{-1}(\widehat{\Lambda})\right).
\end{align*}
Then, for any $\lambda(\cdot)$ satisfying $\lambda(t)\ne0$ for all $t\in[0,T)$ and $\lambda(T-)\ne0$, we have the following assertions:
\begin{enumerate}[label=(\arabic*), font=\upshape]
\item If $\displaystyle\int_0^T|\lambda(t)|^2\md t < \widehat{\Lambda}$, then $|a_1(t)| < |a_2(t)|$ for all $t\in[0,T)$.
    
\item If $\displaystyle\int_0^T|\lambda(t)|^2\md t > \widehat{\Lambda}$, then 
    there exists a unique $t^* \in (0,T)$ such that 
    $|a_1(t)| < |a_2(t)|$ for all $t\in(t^*,T)$ and $|a_1(t)| > |a_2(t)|$ for all $t\in[0,t^*)$. 

\item If $\displaystyle\int_0^T |\lambda(t)|^2 \md t = \widehat{\Lambda}$, then $|a_1(0)| = |a_2(0)|$ and $|a_1(t)| < |a_2(t)|$ for all $t\in (0,T)$.
\end{enumerate}
\begin{proof}
See Appendix \ref{sec:proof:croos}.
\end{proof}
\end{proposition}

\begin{remark}
    It is important to note that the conditions in Propositions \ref{proposition5.2} and \ref{proposition5.3} are compatible with first-order stochastic dominance. One can construct risk aversion distributions such that $\boldsymbol{R}_1 \succeq_{1} \boldsymbol{R}_2$ holds while the conditions in Propositions \ref{proposition5.2} and \ref{proposition5.3} are simultaneously satisfied: letting $p_1=p_2$ in \eqref{eq:star}. For instance, the parameters in Example \ref{exm:fod} satisfy these conditions with $r_0 = 1$, $\delta_1 = 2$, $\delta_2 = 1$, $\mEt[\boldsymbol{R}_1] =1.2>1.1= \mEt[\boldsymbol{R}_2]$ and $\lambda\equiv0.4\neq 0$.
    In particular, Figure \ref{counterexample_plot.png} illustrates both the reversal in risk exposure magnitude and the single-crossing property.
\end{remark}

Having established the single-crossing property, we next investigate the monotonicity of $t^*$ w.r.t. $p \triangleq p_1 = p_2$ under the conditions of Proposition~\ref{proposition5.3}.

\begin{proposition}\label{proposition5.5}
Let $T > 0$ be fixed. 
Assume that the risk aversions $\boldsymbol{R}_1$ and $\boldsymbol{R}_2$ satisfy \eqref{eq:star} with $p_1 = p_2 = p \in (0,1)$ and that $\lambda(t)\ne0$ for all $t\in[0,T)$ with $\lambda(T-)\ne0$. Let $a_{i,p}$ denote the risk exposure vector corresponding to $\boldsymbol{R}_i$ with $p_i=p$ for $i=1,2$. Let $t^*(p)\in(0,T)$ be the crossing time, where $|a_{1,p}(t^*(p))| = |a_{2,p}(t^*(p))|$. Then the set
\begin{align*}
    I \triangleq \left\{ p \in (0,1) : t^*(p) \in (0,T) \right\}
\end{align*}
is open and the crossing time $t^*(\cdot)$ is continuously differentiable in $I$.\footnote{{The set $I$ is non-empty if $\Lambda(0)$ is sufficiently large. Indeed, for any fixed $p \in (0,1)$, Proposition \ref{proposition5.3} ensures that the crossing time $t^*(p)$ exists in $(0,T)$ provided that $\Lambda(0)$ exceeds the threshold $\widehat{\Lambda}$ defined therein, which depends on $p$.}} Moreover, its derivative satisfies 
\begin{align*}
    \frac{\md t^*}{\md p}(p) > 0 \quad \text{for all } p \in I.
\end{align*}
\begin{proof}
See Appendix \ref{append:t^*P}.
\end{proof}
\end{proposition}
 The condition $\dfrac{\md t^*}{\md p} > 0$ indicates that  the crossing time to maturity is strictly decreasing with respect to the  probability parameter $p$. This confirms that a higher common probability of the smallest risk aversion value  brings the crossing closer to maturity.

Building on the analysis in Proposition \ref{proposition5.5}, we now relax the assumption $p_1=p_2$ and investigate the individual effects of the probability parameters $p_i$ ($i=1,2$). The following proposition establishes the comparative statics of the crossing time w.r.t. each investor's probability parameter.
\begin{proposition}\label{proposition5.6}
Let $T > 0$ be fixed. 
Assume that the risk aversions $\boldsymbol{R}_1$ and $\boldsymbol{R}_2$ satisfy \eqref{eq:star} with $\mEt[\boldsymbol{R}_1] > \mEt[\boldsymbol{R}_2]$ and that $\lambda(t)\ne0$ for all $t\in[0,T)$ with $\lambda(T-)\ne0$.
 Let $a_{i,p_i}$ denote the risk exposure vector corresponding to $\boldsymbol{R}_i$ with $\tilde{\mathbb{P}}(\boldsymbol{R}_i=r_0)=p_i$ for $i=1,2$.
Let $t^*(p_1,p_2)\in(0,T)$ denote the crossing time, where $|a_{1,p_1}(t^*(p_1,p_2))| = |a_{2,p_2}(t^*(p_1,p_2))|$. Let
\begin{align*}
    O \triangleq \left\{ (p_1, p_2) \in (0,1)^2 : (1-p_1)\delta_1 > (1-p_2)\delta_2 \text{ and } t^*(p_1, p_2) \in (0,T) \right\}.
\end{align*}
Then, the set $O$ is open\footnote{{The set $O$ is nonempty provided that $\Lambda(0)$ is sufficiently large. The argument is analogous to that in Proposition~\ref{proposition5.5}.}} and the function $t^*(\cdot, \cdot)$ is continuously differentiable in $O$. Moreover, for any $(p_1, p_2) \in O$, the partial derivatives satisfy
\begin{align*}
    \frac{\partial t^*}{\partial p_1}(p_1, p_2) > 0 \quad \text{and} \quad \frac{\partial t^*}{\partial p_2}(p_1, p_2) < 0.
\end{align*}

\begin{proof}
See Appendix \ref{append:t_1*p}.
\end{proof}
\end{proposition}
 Similar to Proposition \ref{proposition5.5}, the differentiability of $t^*$ established in Proposition \ref{proposition5.6} implies that the crossing time varies smoothly with respect to the probability parameters. Specifically, $t^*$ is strictly increasing in $p_1$ and strictly decreasing in $p_2$.  This indicates that a higher probability of the smallest risk aversion for the first investor brings the crossing closer to maturity, whereas for the second investor, it pushes the crossing farther away from maturity.

\subsection{Reverse Hazard Rate Order}

While a higher RRA does not guarantee a less risky investment under first-order stochastic dominance, we demonstrate that it does guarantee such a reduction under the stronger reverse hazard rate order.

We first recall the definition of the reverse hazard rate order (see \cite{caperaa1988tail} and \cite{SS2007}).

\begin{definition}
Let ${\boldsymbol{R}_1}$ and ${\boldsymbol{R}_2}$ be two random variables with cumulative distribution functions $F_{{\boldsymbol{R}_1}}$ and $F_{{\boldsymbol{R}_2}}$, respectively. We say that ${\boldsymbol{R}_1}$ dominates ${\boldsymbol{R}_2}$ in the reverse hazard rate order, denoted ${\boldsymbol{R}_1} \succeq_{rh} {\boldsymbol{R}_2}$, if the ratio $F_{{\boldsymbol{R}_2}}/F_{{\boldsymbol{R}_1}}$
is decreasing\footnote{ Herein, ``increasing" means
``non-decreasing" and ``decreasing" means ``non-increasing."} in  $\{\gamma : F_{{\boldsymbol{R}_1}}(\gamma) > 0\}\cup \{\gamma : F_{{\boldsymbol{R}_2}}(\gamma) > 0\}$. Here, we use the convention $\frac{p}{0}=\infty$ for $p>0$.
\end{definition}

The following proposition establishes a comparison of the risk exposure magnitudes under the reverse hazard rate order. 

\begin{proposition}\label{proposition5.8}
Let ${\boldsymbol{R}_1}$ and ${\boldsymbol{R}_2}$ be two non-negative random variables with $\mEt[{\boldsymbol{R}_i}]<\infty$ ($i=1,2$) and $F_{\boldsymbol{R}_1}\neq F_{\boldsymbol{R}_2}$.
If ${\boldsymbol{R}_1}\succeq_{rh}{\boldsymbol{R}_2}$ and $H_2(\infty) > \Lambda(0)$, then $|a_{1}(t)|\le|a_{2}(t)|$ for any $t\in[0,T)$, where the strict inequality holds if $\lambda(t)\ne0$.
\begin{proof}
See Appendix \ref{Append:rh}. 
\end{proof}
\end{proposition}

We now consider the case where the expectation of risk aversion is infinite. Assume that the optimal equilibria exist for both $\boldsymbol{R}_1$ and $\boldsymbol{R}_2$.  Note that $\displaystyle \varphi(0) = \min\{t \in [0,T] :\,  \Lambda(t)=0\}$ is determined solely by the market parameters and is therefore identical for both investors. Consequently, the comparative result can be established by examining the magnitudes of $a_i(t)$ on the intervals $[0, \varphi(0))$ and $[\varphi(0), T]$ separately.
Specifically, on the interval $[0, \varphi(0))$, the analysis parallels the finite expectation case derived in Proposition \ref{proposition5.8}, and thus the order of the risk exposure magnitude is preserved.
On the interval $[\varphi(0), T)$, the definition of $\varphi(0)$ implies that $\lambda\equiv\textbf{0}$. As a result, the corresponding risk exposure magnitude also satisfies the desired order.
Therefore, we obtain the following corollary.

\begin{corollary}
Let \(\boldsymbol{R}_1\) and \(\boldsymbol{R}_2\) be two non-negative random variables with \(\mEt[\boldsymbol{R}_i] = \infty\) (\(i=1,2\)) and $F_{\boldsymbol{R}_1}\neq F_{\boldsymbol{R}_2}$. Assume that \(\boldsymbol{R}_1 \succeq_{rh} \boldsymbol{R}_2\) and  the optimal equilibria exist for both investors (i.e., \(H_1(\infty) \ge H_2(\infty)> \Lambda(0)\)). Let $a_{i}(\cdot)$ be the risk exposure vector corresponding to the optimal equilibrium strategies. Then 
\(|a_1(t)| \leq |a_2(t)|\) for any $t\in[0,T)$,
where the strict inequality holds  if $\lambda(t) \neq 0$.
\end{corollary}

\subsection{Aggregation of Risk Aversions}
The preceding results focus on comparisons between investments for distinct risk-aversion distributions ordered by stochastic dominance relations. We next study the effect of preference aggregation on the risk exposure magnitude. 
The following proposition shows that  a convex combination of i.i.d. risk aversions leads to a reduction in the risk exposure magnitude.

\begin{proposition}\label{proposition5.10}
Let $\boldsymbol{R}_1, \cdots, \boldsymbol{R}_n$ be i.i.d. non-negative random variables with $\mEt\left[{\boldsymbol{R}_1}\right]<\infty$ and $\displaystyle \boldsymbol{R} \triangleq \sum_{i=1}^n w_i \boldsymbol{R}_i$ be a convex combination with $w_i\in(0,1)$ ($1\leq i\leq n$) and   $\displaystyle \sum_{i=1}^n w_i = 1$. Let $H_1=H_{\boldsymbol{R}_1}$  and $H=H_{\boldsymbol{R}}$. Assume that $H_1(\infty) > \Lambda(0)$.\footnote{Under this assumption, the equilibrium strategy associated with each $\boldsymbol{R}_i$ exists and is unique by Theorem~\ref{theorem3.3}(1). Moreover, the proof in Appendix~\ref{append:convex} shows that $H(\infty)>\Lambda(0)$ also holds.} Let $a(\cdot)$ and $a_1(\cdot)$ denote the risk exposure vector associated with the risk aversions $\boldsymbol{R}$ and $\boldsymbol{R}_1$, respectively. Then
\(|a(t)| \leq |a_1(t)|\)  for all $t \in [0, T)$, where the strict inequality holds 
if $\boldsymbol{R}_1$ is non-degenerate and $\lambda(t)\neq 0$.
\end{proposition}
\begin{proof}
See Appendix \ref{append:convex}.
\end{proof}

It is important to note that the inequality established above relies critically on the i.i.d. assumption. The following example demonstrates that, if independence is retained but the identical distribution assumption is relaxed, the conclusion of Proposition \ref{proposition5.10} may no longer hold.
\begin{example}
Assume that there is only one stock, whose market price of risk is constant, $\lambda(\cdot)\equiv\lambda = 0.5$. The time horizon $T = 50$. 
The risk-aversion distributions of the two investors are:
\begin{itemize}
\item Investor 1: 
$\tilde{\mathbb{P}}(\boldsymbol{R}_1 = 0.1) = 0.2$ and $\tilde{\mathbb{P}}(\boldsymbol{R}_1 = 8) = 0.8$;

\item Investor 2: $\tilde{\mathbb{P}}(\boldsymbol{R}_2 = 1.5) = 1$.
\end{itemize}
Obviously, $\boldsymbol{R}_1$ and $\boldsymbol{R}_2$ are independent. Consider an aggregated investor whose risk aversion is given by $\boldsymbol{R} = 0.5\boldsymbol{R}_1 + 0.5\boldsymbol{R}_2$. 
Figure 2 illustrates the comparative evolution of the risk exposure magnitudes $|a(\cdot)|$,  $|a_1(\cdot)|$, and $|a_2(\cdot)|$. The red-shaded Reversal Region highlights that the investor with the aggregated risk aversion $\boldsymbol{R}$ adopts a strictly more aggressive position than both Investor 1 and Investor 2 when the time to maturity is sufficiently long. 

\begin{figure}[ht]
    \centering
    \includegraphics[width=0.7\textwidth]{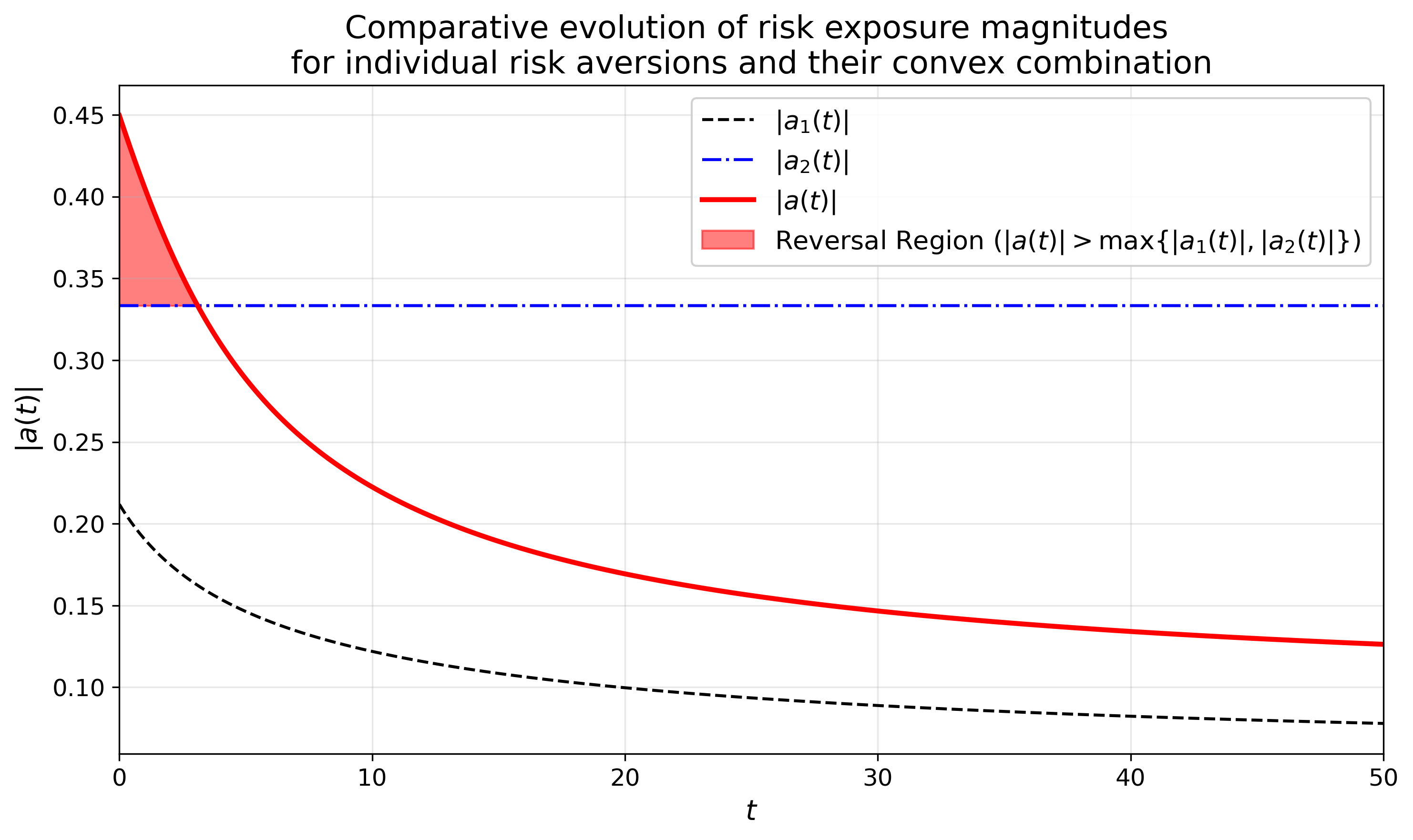}
    \caption{\small Comparative evolution of risk exposure magnitudes for individual risk aversions and their convex combination. 
The black dashed line corresponds to Investor 1 with $\tilde{\mathbb{P}}(\boldsymbol{R}_1 = 0.1) = 0.2$ and $\tilde{\mathbb{P}}(\boldsymbol{R}_1 = 8) = 0.8$; 
the blue dash-dotted line corresponds to Investor 2 with degenerate distribution $\tilde{\mathbb{P}}(\boldsymbol{R}_2 = 1.5) = 1$; 
the red solid line corresponds to the aggregated investor with risk aversion $\boldsymbol{R} = 0.5\boldsymbol{R}_1 + 0.5\boldsymbol{R}_2$. The red-shaded Reversal Region indicates where $|a(t)|>\max\{|a_1(t)|,|a_2(t)|\}$. Parameters: $\lambda = 0.5$,  $T = 50$.}
    \label{counterexample_plot_2.png}
\end{figure}

\end{example}

Proposition~\ref{proposition5.10} establishes a comparison result for convex combinations of i.i.d. risk aversions with a fixed population size, but does not address the asymptotic behavior as the number of samples increases. We therefore focus on the special case of the sample mean risk aversion. The following proposition examines the limit of the risk exposure vector as the sample size 
$n$ tends to infinity and establishes monotonic convergence by exploiting the structure of the sample mean.

\begin{proposition}\label{prop:lln}
Suppose that $\{\boldsymbol{R}_i\}_{i=1}^\infty$ is a sequence of i.i.d. non-negative random variables with $\mEt\left[{\boldsymbol{R}_1}\right]=\mu \in (0,\infty)$. Let $\displaystyle \boldsymbol{Y}_n \triangleq \frac{1}{n} \sum_{i=1}^n \boldsymbol{R}_i$ denote the sample mean risk aversion. Let $H_1$ (resp. $\tilde H_n$) denote  the function defined in \eqref{eq:H} corresponding to $\boldsymbol{R}_1$(resp. $\boldsymbol{Y}_n$). Assume that $H_1(\infty) > \Lambda(0)$. Then the deterministic equilibrium strategies associated with each ${\boldsymbol{R}_i}$ and $\boldsymbol{Y}_n$ exist   and are unique. Let $\tilde a_n(\cdot)$ be the risk exposure vector associated to the risk aversion $\boldsymbol{Y}_n$. Then, as $n \to \infty$, we have
\begin{align}\label{conv:an}
 |\tilde a_n(t)| \searrow \frac{|\lambda(t)|}{\mu} \quad \text{for all }t \in [0, T).
\end{align}
\begin{proof}
See Appendix \ref{proof:prop:lln}.
\end{proof}
\end{proposition}

\section*{Acknowledgement}
   Zongxia Liang acknowledges the financial support from the National Natural Science Foundation of
China (grant 12271290).
Sheng Wang acknowledges Professor Ka Chun Cheung and the financial supports as a postdoctoral fellow from the Department of Statistics and Actuarial Science, School of Computing and Data Science, the University of Hong Kong. 
Jianming Xia acknowledges the financial support from the
National Natural Science Foundation of China (grants  12431017 and 12471447).

\begin{appendices}
\section{Properties of the function $h$}\label{appendix A}

In this appendix, we present some properties of the function $h$ defined in \eqref{eq:h}.

\begin{lemma} [Monotonicity]\label{h:mono}
The function $h$ is increasing in $[0,\infty)$. The monotonicity is strict unless $\boldsymbol{R}$ is a constant.
\end{lemma}
\begin{proof}
By direct differentiation, we have
\begin{align*}
    h'(x)= \frac{1}{2} \frac{\mEt\left[e^{-\frac{1}{2}\boldsymbol{R}x}\right] \mEt\left[\boldsymbol{R}^2 e^{-\frac{1}{2}\boldsymbol{R}x}\right] - \left(\mEt\left[\boldsymbol{R} e^{-\frac{1}{2}\boldsymbol{R}x}\right]\right)^2}{\left(\mEt\left[\boldsymbol{R} e^{-\frac{1}{2}\boldsymbol{R}x}\right]\right)^2}\quad \text{for all }x>0.
\end{align*}
By the Cauchy-Schwartz inequality,
\begin{align}\label{ineq:holder}
    \left(\mEt\left[\boldsymbol{R} e^{-\frac{1}{2}\boldsymbol{R}x}\right]\right)^2
    \leq\mEt\left[e^{-\frac{1}{2}\boldsymbol{R}x}\right] \mEt\left[\boldsymbol{R}^2 e^{-\frac{1}{2}\boldsymbol{R}x}\right].
\end{align}
 Hence, $h'\ge0$ and $h$ is increasing. The equality in \eqref{ineq:holder} holds if and only if $\mEt\left[e^{-\frac{1}{2}\boldsymbol{R}x}\right]\boldsymbol{R}^2 e^{-\frac{1}{2}\boldsymbol{R}x}=e^{-\frac{1}{2}\boldsymbol{R}x}\mEt\left[\boldsymbol{R}^2 e^{-\frac{1}{2}\boldsymbol{R}x}\right]$, which is equivalent to $\boldsymbol{R}$ following a single-point distribution. Therefore, $h$ is strictly increasing unless $\boldsymbol{R}$ is a constant.
 \end{proof}

\begin{lemma}[Asymptotic Analysis]\label{lemmaA.3}
Let $\boldsymbol{R}$ be a non-negative random variable and let $r_0 \triangleq \operatorname{essinf} \boldsymbol{R}$. Then
$\lim\limits_{x \to \infty} h(x) = \frac{1}{r_0}$,
where we adopt the convention $\frac{1}{0} = \infty$.
\begin{proof} Let $\boldsymbol{X} \triangleq \boldsymbol{R} - r_0$. Then 
\begin{align}\label{eq:h and Q}
\frac{1}{h(x)} = \frac{\mEt[\boldsymbol{R} e^{-\frac{1}{2}\boldsymbol{R} x}]}{\mEt[e^{-\frac{1}{2}\boldsymbol{R} x}]} = \frac{\mEt[(r_0 + \boldsymbol{X}) e^{-\frac{1}{2}(r_0 + \boldsymbol{X}) x}]}{\mEt[e^{-\frac{1}{2}(r_0 + \boldsymbol{X}) x}]} = r_0 + Q(x),
\end{align}
where 
$Q(x) \triangleq \frac{\mEt[\boldsymbol{X} e^{-\frac{1}{2}\boldsymbol{X} x}]}{\mEt[e^{-\frac{1}{2}\boldsymbol{X} x}]}$.
 By $\operatorname{essinf} \boldsymbol{X} = 0$, we have $p_\epsilon \triangleq \tilde{\mathbb{P}}(0 \le \boldsymbol{X} \le \epsilon) > 0$ for any $\epsilon>0$. Therefore,
\begin{align*}
Q(x)\leq & \frac{\mEt[\boldsymbol{X} e^{-\frac{1}{2}\boldsymbol{X} x} \mathbb{I}_{\{\boldsymbol{X}\leq \epsilon\}}+\mEt[\boldsymbol{X} e^{-\frac{1}{2}\boldsymbol{X} x} \mathbb{I}_{\{\boldsymbol{X} > \epsilon\}}]}{\mEt[e^{-\frac{1}{2}\boldsymbol{X} x} \mathbb{I}_{\{\boldsymbol{X} \le  \epsilon\}}]}\\
 \leq &\epsilon+ \frac{\mEt[\boldsymbol{X} e^{-\frac{1}{2}\boldsymbol{X} x} \mathbb{I}_{\{\boldsymbol{X} > \epsilon\}}]}{p_\epsilon e^{-\frac{1}{2}\epsilon x}} =\epsilon+ \frac{\mEt[\boldsymbol{X} e^{-\frac{1}{2}(\boldsymbol{X}-\epsilon) x} \mathbb{I}_{\{\boldsymbol{X} > \epsilon\}}]}{p_\epsilon} .
\end{align*}
Note that for any $x \ge 1$, $\boldsymbol{X} e^{-\frac{1}{2}(\boldsymbol{X}-\epsilon) x} \mathbb{I}_{\{\boldsymbol{X} > \epsilon\}}$ is dominated by $\boldsymbol{X} e^{-\frac{1}{2}(\boldsymbol{X}-\epsilon)} \mathbb{I}_{\{\boldsymbol{X} > \epsilon\}}\in L^\infty$. 
Hence, by the dominated convergence theorem, $\limsup\limits_{x \to \infty} Q(x) \le \epsilon$. Because $\epsilon$ is arbitrary, we have $\lim\limits_{x \to \infty} Q(x) = 0$.  By \eqref{eq:h and Q}, the convergence result for $h$ follows immediately.
\end{proof}
\end{lemma}

\begin{lemma}[Uniform Convergence]\label{lma:h_ntoh}
Let $\{\boldsymbol{R}_n\}_{n\ge 1}$ be a sequence of nonnegative random variables that converges in distribution to $\boldsymbol{R}$ and $\lim\limits_{n \to \infty} \mEt[\boldsymbol{R}_n] = \mEt[\boldsymbol{R}] \in (0, \infty)$.
Then, on every compact subset of $[0, \infty)$, the sequence $\{h_{\boldsymbol{R}_n}\}_{n\ge1}$ converges uniformly to $h_{\boldsymbol{R}}$.

\begin{proof}
Fix a compact subset $\mathcal K\subset[0,\infty)$.

We first show that $\mEt[e^{-\frac{1}{2}\boldsymbol{R}_n x}] \to \mEt[e^{-\frac{1}{2}\boldsymbol{R} x}]$ uniformly on $\mathcal K$.
Observe that
\begin{align*}
\left|\dfrac{\md\left(\mEt[e^{-\frac{1}{2}\boldsymbol{R}_n x}]\right)}{\md x} \right| = \left| -\frac{1}{2}\mEt[\boldsymbol{R}_n e^{-\frac{1}{2}\boldsymbol{R}_n x}] \right| \le \frac{1}{2}\mEt[\boldsymbol{R}_n]\quad \text{for all}\,\,x\geq0.
\end{align*}
The sequence $\{\mEt[\boldsymbol{R}_n]\}_{n\geq 1}$ is bounded because $\mEt[\boldsymbol{R}_n] \to \mEt[\boldsymbol{R}]$. Hence, the family of functions $\{\mEt[e^{-\frac{1}{2}\boldsymbol{R}_n x}]\}_{n\geq1}$ is equicontinuous on $\mathcal{K}$. Moreover, we have $\mEt[e^{-\frac{1}{2}\boldsymbol{R}_n x}] \to \mEt[e^{-\frac{1}{2}\boldsymbol{R} x}]$ for any $x\geq0$, because $\boldsymbol{R}_n \xrightarrow{d} \boldsymbol{R}$ and the function
$r\mapsto e^{-\frac12 r x}$ is bounded and continuous. It follows that $\mEt[e^{-\frac{1}{2}\boldsymbol{R}_n x}] \to \mEt[e^{-\frac{1}{2}\boldsymbol{R} x}]$ uniformly on $\mathcal K$;  see, e.g., \citet[Exercise 7.16]{rudin1976principles}.
 
Now we show the uniform convergence of $\{\mEt[\boldsymbol{R}_n e^{-\frac{1}{2}\boldsymbol{R}_n x}]\}_{n\ge1}$ on $\mathcal K$. By
$\boldsymbol{R}_n \xrightarrow{d} \boldsymbol{R}$ and $\mEt[\boldsymbol{R}_n] \to \mEt[\boldsymbol{R}] < \infty$,  we know that the family $\{\boldsymbol{R}_n\}_{n\geq 1}$ is uniformly integrable. Consequently, for any $\epsilon>0$, there exists $M>0$ such that $F_{\boldsymbol{R}}$ is continuous at $M$, $\displaystyle\sup_n \mEt[\boldsymbol{R}_n \mathbb{I}_{\{\boldsymbol{R}_n > M\}}]<\frac{\epsilon}{3}$, and $\displaystyle\mEt[\boldsymbol{R} \mathbb{I}_{\{\boldsymbol{R} > M\}}]<\frac{\epsilon}{3}$. Let $f_n^M(x) \triangleq \mEt[\boldsymbol{R}_n e^{-\frac{1}{2}\boldsymbol{R}_n x} \mathbb{I}_{\{\boldsymbol{R}_n \le M\}}]$. Then $f_n^M$ is continuously differentiable and its derivative is uniformly bounded by $\frac12 M^2$, which implies that the family $\{f_n^M\}_{n\geq 1}$ is equicontinuous on $\mathcal{K}$. Furthermore, because $\boldsymbol{R}_n \xrightarrow{d} \boldsymbol{R}$, we have the pointwise convergence $f_n^M(x)\to f^M(x) \triangleq\mEt[\boldsymbol{R} e^{-\frac{1}{2}\boldsymbol{R} x} \mathbb{I}_{\{\boldsymbol{R} \le M\}}]$; see, e.g., \citet[Theorem A.43]{follmer2016stochastic}. As a result, $f_n^M$ converges uniformly to $f^M$ on $\mathcal K$. Hence, for all sufficiently large $n$,
$$\sup_{x \in \mathcal{K}} |f_n^M(x) - f^M(x)| < \frac{\epsilon}{3}.$$ Observe that
\begin{align*}
\left|\mEt[\boldsymbol{R}_n e^{-\frac{1}{2}\boldsymbol{R}_n x}] - \mEt[\boldsymbol{R} e^{-\frac{1}{2}\boldsymbol{R} x}]\right| \le |f_n^M(x) - f^M(x)| + \mEt[\boldsymbol{R}_n \mathbb{I}_{\{\boldsymbol{R}_n > M\}}] + \mEt[\boldsymbol{R} \mathbb{I}_{\{\boldsymbol{R} > M\}}].
\end{align*}
Therefore, for $n$ large enough, $$\sup_{x \in \mathcal{K}} \left|\mEt[\boldsymbol{R}_n e^{-\frac{1}{2}\boldsymbol{R}_n x}] - \mEt[\boldsymbol{R} e^{-\frac{1}{2}\boldsymbol{R} x}]\right| < \epsilon,$$
 which implies the uniform convergence of $\{\mEt[\boldsymbol{R}_n e^{-\frac{1}{2}\boldsymbol{R}_n x}]\}_{n\ge1}$ on $\mathcal K$. 
 
 Finally, the function $\mEt[\boldsymbol{R} e^{-\frac{1}{2}\boldsymbol{R} x}]$ is bounded away from zero on $\mathcal{K}$, given that $\mEt[\boldsymbol{R}]>0$. Therefore,  $\{h_{\boldsymbol{R}_n}\}_{n\ge1}$ converges uniformly to $h_{\boldsymbol{R}}$ on $\mathcal K$.
\end{proof}
\end{lemma}

The reverse hazard rate order can be characterized via expectation ratios, as established by \cite{caperaa1988tail} (see also \citet[Theorem 1.B.50]{SS2007}). 

\begin{lemma}\label{lemma5.4}
Let ${\boldsymbol{R}_1}$ and ${\boldsymbol{R}_2}$ be two random variables. Then ${\boldsymbol{R}_1} \succeq_{rh} {\boldsymbol{R}_2}$ if and only if
\begin{align*}
\frac{\mEt\left[u({\boldsymbol{R}_1})w({\boldsymbol{R}_1})\right]}{\mEt[w({\boldsymbol{R}_1})]} \ge \frac{\mEt[u({\boldsymbol{R}_2})w({\boldsymbol{R}_2})]}{\mEt[w({\boldsymbol{R}_2})]}
\end{align*}
for all functions $u$ and $w$ such that the expectations exist, $w$ is positive and decreasing, and $u$ is  increasing.
\end{lemma}

The next lemma shows that $h_{\boldsymbol{R}}$ is ``decreasing" w.r.t. $\boldsymbol{R}$ (under the reverse hazard rate order of $\boldsymbol{R}$).

\begin{lemma}\label{proposition5.7} 
Let $\boldsymbol{R}_1$ and $\boldsymbol{R}_2$ be two non-negative random variables with $F_{{\boldsymbol{R}_1}} \neq F_{{\boldsymbol{R}_2}}$. If $\boldsymbol{R}_1 \succeq_{rh}\boldsymbol{R}_2 $, then $h_1(x) < h_2(x)$ for all $x\ge 0$,
where $h_i=h_{\boldsymbol{R}_i}$ ($i=1, 2$).

\begin{proof}
Fix $x \ge 0$. Let $w(\gamma) = \exp(-\frac{1}{2}\gamma x)$ and $u(\gamma) = \gamma$. Applying Lemma \ref{lemma5.4} yields
\begin{align} \label{eq:ratio_ineq}
    \frac{\mEt[{\boldsymbol{R}_1} w({\boldsymbol{R}_1})]}{\mEt[w({\boldsymbol{R}_1})]} \ge \frac{\mEt[{\boldsymbol{R}_2} w({\boldsymbol{R}_2})]}{\mEt[w({\boldsymbol{R}_2})]},
\end{align}
which implies that $h_1(x) \leq h_2(x)$.
To show the strict inequality, suppose on the contrary that the equality holds in (\ref{eq:ratio_ineq}). Let $\tilde{F}_i$ ($i=1,2$) be the probability distribution functions defined by $\md \tilde{F}_i(\gamma) = \frac{w(\gamma)\md F_{\boldsymbol{R}_i}(\gamma)}{\mEt[w(\boldsymbol{R}_i)]}$. 
Then $\int_0^\infty \gamma \md \tilde{F}_1(\gamma)=\int_0^\infty \gamma \md \tilde{F}_2(\gamma)$.
On the other hand,  a combination of $\boldsymbol{R}_1 \succeq_{rh} \boldsymbol{R}_2$, Lemma \ref{lemma5.4}, and the definition of $\tilde{F}_i$ yields $\tilde{F}_1 \succeq_1 \tilde{F}_2$. 
Therefore, 
 $\tilde{F}_1 = \tilde{F}_2$. 
 By $w(\cdot) > 0$, we have $F_{{\boldsymbol{R}_1}} = F_{{\boldsymbol{R}_2}}$, which contradicts the assumption that $F_{{\boldsymbol{R}_1}} \neq F_{{\boldsymbol{R}_2}}$. Thus, the inequality must be strict.
\end{proof}
\end{lemma}

\section{Proofs}\label{appendix B}

\subsection{Proof of Proposition \ref{proposition3.10}}\label{proof of proposition3.10}
\paragraph{(1)} Assume that the sequence $\{\Lambda_n\}_{n \ge 1}$ strictly increases to $H(\infty)$. First, we have that $\Lambda_n < H(\infty)$ for all $n \ge 1$, which guarantees the existence of the equilibrium.
Then, for each $n$, $v_n(0)$ is well-defined and uniquely determined by $v_n(0) = H^{-1}(\Lambda_n)$. Consequently, the convergence $\Lambda_n \to H(\infty)$ implies that $\displaystyle\lim_{n \to \infty} v_n(0)= \infty$.
Based on Theorem \ref{theorem3.3}, we have $a_n(0) = h(v_n(0))\lambda_0$. Using the asymptotic property of $h$ established in Lemma \ref{lemmaA.3}, i.e., $\displaystyle\lim_{x \to \infty} h(x) = \frac{1}{r_0}$, we obtain
\begin{align*}
\lim_{n \to \infty} a_n(0) = \lambda_0 \lim_{v_n(0) \to \infty} h(v_n(0)) = \frac{\lambda_0}{r_0}.
\end{align*}
Furthermore, $\{v_n(0)\}$ is strictly increasing as $\{\Lambda_n\}$ is strictly increasing. Combined with the monotonicity of $h$ (see Lemma \ref{h:mono}), it follows that $|a_n(0)| = |\lambda_0|h(v_n(0))$ increases to $\frac{|\lambda_0|}{r_0}$.  

\paragraph{(2)} Assume that the sequence $\{\Lambda_n\}_{n \ge 1}$ strictly decreases to $0$. By $H(\infty) > 0$,  $\Lambda_n < H(\infty)$ for all sufficiently large $n$. Then the equilibrium exists for all sufficiently large $n$.
For such $n$, we have $v_n(0) = H^{-1}(\Lambda_n)$. Because $H(0)=0$ and $H^{-1}$ is continuous, we have
\begin{align*}
\lim_{n \to \infty} v_n(0) = H^{-1}(0) = 0.
\end{align*}
 Recall that $\displaystyle h(0) = \frac{1}{\mEt[\boldsymbol{R}]}$. 
Taking the limit of $a_n(0)$ yields
\begin{align*}
\lim_{n \to \infty} a_n(0) = \lambda_0 \lim_{v_n(0) \to 0} h(v_n(0)) = \frac{\lambda_0}{\mEt[\boldsymbol{R}]}.
\end{align*}
Finally,  $\{v_n(0)\}$ is strictly decreasing as $\{\Lambda_n\}$ is strictly decreasing. The strict monotonicity of $h$ implies that $|a_n(0)|$ decreases to $\displaystyle \frac{|\lambda_0|}{\mEt[\boldsymbol{R}]}$.

\subsection{Proof of Theorem \ref{thm:converge:distribution}}\label{proof:thm:converge:distribution}

For notational simplicity, let $h_n=h_{\boldsymbol{R}_n}$, $H_n=H_{\boldsymbol{R}_n}$, and $v_n=v_{\boldsymbol{R}_n}$ ($n\ge1$).
By Lemma \ref{lma:h_ntoh}, the sequence
$\{h_n\}_{n\ge1}$ converges uniformly to $h$ on every compact subset of $[0,\infty)$. Because the limit function $h$ is strictly positive, it follows that $\{H_n\}_{n\ge1}$ converges uniformly to $H$ on every compact subset of $[0,\infty)$ as well.  By Lemma \ref{lemmaB.1} below, the sequence of the inverse functions $\{H_n^{-1}\}_{n\ge1}$ converges uniformly to $H^{-1}$ on $[0,\Lambda(0)]$.

We now see the uniform convergence of $\{v_n\}_{n\ge1}$ to $v$ on $[0,T]$:
\begin{align*}
\sup_{t \in [0, T]} |v_n(t) - v(t)| = \sup_{t \in [0, T]} |H_n^{-1}(\Lambda(t)) - H^{-1}(\Lambda(t))| \le \sup_{y \in [0, \Lambda(0)]} |H_n^{-1}(y) - H^{-1}(y)| \to 0.
\end{align*}

Next, we show the uniform convergence of the composite sequence $h_n(v_n(\cdot))$. Indeed, by the triangle inequality,
\begin{align}\label{eq:a.1}
\sup_{t \in [0, T]} |h_n(v_n(t)) - h(v(t))| \le \sup_{t \in [0, T]} |h_n(v_n(t)) - h(v_n(t))| + \sup_{t \in [0, T]} |h(v_n(t)) - h(v(t))|.
\end{align}
Let $\displaystyle\mathcal{K} \triangleq \left[0,\max_{t\in[0,T]}|v(t)|+1\right]$. Then  $v_n(t)\in \mathcal{K}$ for all $t\in[0,T]$  and all sufficiently large $n$.
Therefore, the uniform convergence of $\{h_n\}_{n\ge1}$ on $\mathcal{K}$ implies $\displaystyle\sup_{t \in [0, T]} |h_n(v_n(t)) - h(v_n(t))| \to 0$. Moreover, $h$ is continuous and therefore uniformly continuous on the compact set $\mathcal{K}$, which combined with the fact that $\sup\limits_{t \in [0, T]} |v_n(t) - v(t)| \to 0$ implies $\sup\limits_{t \in [0, T]} |h(v_n(t)) - h(v(t))| \to 0$. Consequently, by \eqref{eq:a.1}, we have
\begin{align}\label{eq:h_uniform}
\lim_{n \to \infty} \sup_{t \in [0, T]} |h_n(v_n(t)) - h(v(t))| = 0.
\end{align}

Finally, recalling that $a_n(t) = h_n(v_n(t))\lambda(t)$ and $a(t) = h(v(t))\lambda(t)$, we get from \eqref{eq:h_uniform} that $\lim\limits_{n \to \infty} a_n(t) = a(t)$  for every $t\in[0,T)$.
Furthermore, if  $\lambda \in L^\infty([0, T))$, then the above convergence is uniform on $[0,T)$. 

\begin{lemma}\label{lemmaB.1}

Let $H$ and $H_n$ ($n\ge 1$) be strictly increasing and continuous functions defined on $[0, \infty)$ with $H(0)=H_n(0)=0$ for all $n\ge 1$. 
Suppose that $\{H_n\}_{n\ge 1}$ converges uniformly to $H$ on every compact subset of $[0, \infty)$. 
Then, for any bounded interval $[0, z] \subset [0, H(\infty))$, there exists an integer $N$ such that for all $n \ge N_0$, the inverse function $H_n^{-1}$ is well-defined on $[0, z]$, and the sequence $\{H_n^{-1}\}_{n\ge N_0}$ converges uniformly to $H^{-1}$ on $[0, z]$.
\begin{proof} 
Fix $[0, z]\subset[0,H(\infty))$ and a constant $M>H^{-1}(z)$. By the uniform convergence of $\{H_n\}_{n\ge1}$ on $[0,M]$, there exists $N_0\ge1$ such that $H_n(M)>z$ for all $n\ge N_0$. Because each $H_n$ is strictly increasing and continuous with $H_n(0)=0$, it follows that $[0, z] \subset [0, H_n(M)] = H_n([0, M])$ for all $n\geq N_0$. Thus, for all $n \ge N_0$, the inverse function $H_n^{-1}$ is well-defined on $[0, z]$, and moreover, for any $x \in [0, z]$, we have $H_n^{-1}(x) \in [0, M]$ and $H^{-1}(x)\in[0,M]$.

Next, because $H^{-1}$ is continuous on the compact interval $[0, H(M)]$, it is uniformly continuous. Thus, for any $\epsilon > 0$, there exists  a $\delta > 0$ such that for any $x', x'' \in [0, H(M)]$ with $|x' - x''| < \delta$, we have 
\begin{align}\label{eq:unif_cont}
    |H^{-1}(x') - H^{-1}(x'')| < \epsilon. 
\end{align}
By uniform convergence of $H_n \to H$ on $[0, M]$, there exists  a $N \ge N_0$ such that for all $n \ge N$ and all $y \in [0, M]$,
\begin{align*}
    |H_n(y) - H(y)| < \delta.
\end{align*}
For any $x \in [0, z]$, let $x_n = H_n^{-1}(x)$. Then,  $x_n \in [0, M]$ and 
\begin{align*}
    |H(x_n) - x| = |H(x_n) - H_n(x_n)| < \delta.
\end{align*}
Because $x, H(x_n) \in [0, H(M)]$ and $x_n = H^{-1}(H(x_n))$, the inequality (\ref{eq:unif_cont}) yields
\begin{align*}
 |H_n^{-1}(x)-H^{-1}(x)|=   |x_n - H^{-1}(x)| = |H^{-1}(H(x_n)) - H^{-1}(x)| < \epsilon.
\end{align*}
As this holds uniformly for all $x \in [0, z]$, the convergence is uniform.
\end{proof}
\end{lemma}

\subsection{Proof of Proposition \ref{proposition5.2}}\label{proof of proposition5.2}
We first establish a technical lemma.
\begin{lemma} \label{lemma5.3}
    Let $\boldsymbol{R}_1$ and $\boldsymbol{R}_2$ be two non-negative random variables satisfying Conditions (c1) and (c2) in Proposition \ref{proposition5.2}.  Following the notation in the proof of Lemma \ref{lemmaA.3}, let
    \begin{align}\label{eq:Xi:Qi}
    \boldsymbol{X}_i = \boldsymbol{R}_i - r_0\text{ and }
    Q_i(x) = \frac{\mEt[\boldsymbol{X}_i e^{-\frac{1}{2}\boldsymbol{X}_i x}]}{\mEt[e^{-\frac{1}{2}\boldsymbol{X}_i x}]},\quad x\geq 0.
    \end{align}
    Then we have the following assertions.
    \begin{enumerate}[label=(\arabic*), font=\upshape]
        \item For each $i\in\{1,2\}$, there exists a constant $M_i>0$ such that $Q_i(x) \le M_i e^{-\frac{1}{2}\delta_i x}$ for  $x\geq 0$. In particular, $Q_i(x) \to 0$ as $x \to \infty$.
        \item For each $i\in\{1,2\}$, the integral $K_i=\displaystyle\int_{0}^{\infty} (h_i^{-2}(s) - r_0^2) \md s$ is finite.
        \item For any $\varepsilon > 0$, there exists a constant $C > 0$ such that $\dfrac{Q_1(x_1)}{Q_2(x_2)} \le C\dfrac{e^{-\frac{1}{2}\delta_1 x_1}}{ e^{-\frac{1}{2}(\delta_2 + \varepsilon) x_2}}$ for any $x_1, x_2 \ge 0$.
        
    \end{enumerate}
\begin{proof}\
\paragraph{(1)}
For $i\in\{1,2\}$.
By Condition (c1) in Proposition~\ref{proposition5.2}, we have $ \tilde{\mathbb{P}} \mathbb(\boldsymbol{X}_i=0)=p_i$. The dominated convergence theorem then implies $\lim\limits_{x \to \infty} \mEt[e^{-\frac{1}{2}\boldsymbol{X}_i x}] = p_i\in(0,1)$. Hence, there exists $m_i > 0$ such that $\mEt[e^{-\frac{1}{2}\boldsymbol{X}_i x}]\geq m_i$ for any $x\geq0$. 
    Moreover, by Condition (c2) in Proposition \ref{proposition5.2}, the support of $\boldsymbol{X}_i$ on $(0, \infty)$ is contained in $[\delta_i, \infty)$. As such,
    \begin{align*}
    \mEt[\boldsymbol{X}_i e^{-\frac{1}{2}\boldsymbol{X}_i x}] = \int_{\delta_i}^{\infty} \gamma e^{-\frac{1}{2}\gamma x} \md F_{\boldsymbol{X}_i}(\gamma) \le e^{-\frac{1}{2}\delta_i x} \int_{\delta_i}^{\infty} \gamma \md F_{\boldsymbol{X}_i}(\gamma) = e^{-\frac{1}{2}\delta_i x} \mEt[\boldsymbol{X}_i \mathbb{I}_{\{\boldsymbol{X}_i \ge \delta_i\}}].
    \end{align*}
    Thus the desired conclusion  follows by taking $M_i\triangleq \frac{\mEt[\boldsymbol{X}_i \mathbb{I}_{\{\boldsymbol{X}_i \ge \delta_i\}}]}{m_i}$.

\paragraph{(2)}
   For $i\in\{1,2\}$, based on \eqref{eq:h and Q}, the conclusion follows immediately from Assertion (1) and the fact that $\delta_i>0$.

\paragraph{(3)}
    For any $\varepsilon > 0$. By the definition of the support, we have $p_{\varepsilon,2} \triangleq P(\delta_2 \le \boldsymbol{X_2} < \delta_2 + \varepsilon) > 0$. Restricting the integration to this interval yields
    \begin{align*}
    \mEt[\boldsymbol{X_2} e^{-\frac{1}{2}\boldsymbol{X_2} x}] \ge \int_{\delta_2}^{\delta_2 + \varepsilon} x e^{-\frac{1}{2}x \gamma} \md F_{\boldsymbol{X_2}}(\gamma) \ge \delta_2 e^{-\frac{1}{2}(\delta_2 + \varepsilon)x} p_{\varepsilon,2}.
    \end{align*}
   Combining this with Assertion (1) gives the desired bound with $C\triangleq \frac{M_1}{\delta_2 p_{\varepsilon,2}}$.
    \end{proof}
\end{lemma}
Now we return to prove Proposition \ref{proposition5.2}.

By Assertion (1) of Lemma~\ref{lemma5.3} and \eqref{eq:h and Q}, we have $\displaystyle\lim_{x\to\infty} h_i(x) = \frac{1}{r_0}$. Because $r_0 > 0$, it follows that $\int_{0}^{\infty} h_i^{-2}(x) \md x=\infty$, which implies $H_i(\infty) = \infty$. Consequently, Theorem~\ref{theorem3.3}(1) guarantees the existence and uniqueness of the deterministic equilibrium strategy corresponding to $\boldsymbol{R}_i$ for any deterministic and right-continuous market price of risk $\lambda(\cdot)$ in $L^2([0,T))$.

First, we establish the asymptotic relationship between the inverse functions $H_1^{-1}(\Lambda)$ and $H_2^{-1}(\Lambda)$ for large scalar $\Lambda$.  
Because $H_i(\infty) = \infty$, it follows that  $H_i^{-1}(\Lambda) \to \infty$ as $\Lambda \to \infty$. Moreover, the identity $\Lambda = H_i(H_i^{-1}(\Lambda))$ yields
\begin{align}\label{eq:Psi(0)}
\Lambda=\int_{0}^{H_i^{-1}(\Lambda)} (r_0^2 + (h_i^{-2}(x) - r_0^2)) \md x = r_0^2 H_i^{-1}(\Lambda) + \int_{0}^{H_i^{-1}(\Lambda)} (h_i^{-2}(x) - r_0^2) \md x.
\end{align}
Define the tail integral $\Theta_i(u) \triangleq \int_{u}^{\infty} (h_i^{-2}(x) - r_0^2) \md x$, which satisfies $\displaystyle\lim_{u \to \infty} \Theta_i(u) = 0$ based on Lemma \ref{lemma5.3}(2). Then \eqref{eq:Psi(0)} can be rewritten as
\begin{align}\label{eq:v(0)}
\Lambda = r_0^2 H_i^{-1}(\Lambda) + K_i - \Theta_i(H_i^{-1}(\Lambda)).
\end{align}
Equating (\ref{eq:v(0)}) for $i=1$ and $i=2$ gives
\begin{align*}
H_1^{-1}(\Lambda) = H_2^{-1}(\Lambda) + \dfrac{K_2 - K_1}{r_0^2} + \dfrac{\Theta_2(H_2^{-1}(\Lambda)) - \Theta_1(H_1^{-1}(\Lambda))}{r_0^2}.
\end{align*}
Let $\Delta K \triangleq \dfrac{K_2 - K_1}{r_0^2}$ and  $\xi(\Lambda) \triangleq \dfrac{\Theta_2(H_2^{-1}(\Lambda)) - \Theta_1(H_1^{-1}(\Lambda))}{r_0^2}$. Because $H_i^{-1}(\Lambda) \to \infty$ as $\Lambda \to \infty$, the tail difference term $\xi(\Lambda)$ tends to 0 as $\Lambda \to \infty$.

Second, we compare $h_1(H_1^{-1}(\Lambda))$ and $h_2(H_2^{-1}(\Lambda))$. Because $h_i(s) = (r_0 + Q_i(s))^{-1}$, this reduces to comparing $Q_1(H_1^{-1}(\Lambda))$ and $Q_2(H_2^{-1}(\Lambda))$. Let $\varepsilon > 0$ be a constant satisfying $\delta_1 - \delta_2 - \varepsilon > 0$.
Using the bound established in Lemma~\ref{lemma5.3}(3), we obtain
\begin{align*}
\dfrac{Q_1(H_1^{-1}(\Lambda))}{Q_2(H_2^{-1}(\Lambda))} \le C \dfrac{ e^{-\frac{1}{2}\delta_1 H_1^{-1}(\Lambda)}}{ e^{-\frac{1}{2}(\delta_2 + \varepsilon) H_2^{-1}(\Lambda)}}.
\end{align*}
Substituting $H_1^{-1}(\Lambda) = H_2^{-1}(\Lambda) + \Delta K + \xi(\Lambda)$ into the exponent yields
\begin{align*}
\dfrac{Q_1(H_1^{-1}(\Lambda))}{Q_2(H_2^{-1}(\Lambda))}\le C e^{-\frac{1}{2}\delta_1 (\Delta K + \xi(\Lambda))} \cdot e^{-\frac{1}{2}(\delta_1 - \delta_2 - \varepsilon) H_2^{-1}(\Lambda)}.
\end{align*}
As $\Lambda \to \infty$, we have $\xi(\Lambda) \to 0$, so the first exponential term converges to a finite and positive constant. Moreover, $H_2^{-1}(\Lambda) \to \infty$, so the second term tends to $0$. Hence, the ratio $\frac{Q_1(H_1^{-1}(\Lambda))}{Q_2(H_2^{-1}(\Lambda))}$ tends to 0 as $\Lambda \to \infty$.
It follows that there exists a constant $\overline \Lambda > 0$ such that, for all $\Lambda > \overline \Lambda$, we have $Q_1(H_1^{-1}(\Lambda)) < Q_2(H_2^{-1}(\Lambda))$ and thus
\begin{align}\label{eq:h_inequality}
h_1(H_1^{-1}(\Lambda)) > h_2(H_2^{-1}(\Lambda)) \quad \text{for all } \Lambda > \overline \Lambda.
\end{align}

Next, we compare the risk exposure magnitudes at time $0$. For any market price of risk $\lambda(\cdot)$ satisfying $\lambda(0)\neq 0$ and $\Lambda(0) \triangleq \int_0^T |\lambda(t)|^2 \, \md t > \overline{\Lambda}$. We have $v_i(0) = H_i^{-1}(\Lambda(0))$ according to Theorem~\ref{theorem3.3}(1). Applying \eqref{eq:h_inequality} with $\Lambda = \Lambda(0)$, we obtain $h_1(v_1(0)) > h_2(v_2(0))$ because $\Lambda(0) > \overline{\Lambda}$. Then, we conclude that $|a_1(0)| > |a_2(0)|$ because the risk exposure magnitude satisfies $|a_i(0)| = h_i(v_i(0))|\lambda(0)|$ by (\ref{eq:equilibrium:g}).

Finally, because $v_i(T)=0$, we have $h_i(v_i(T))=\dfrac{1}{\mEt[R_i]}$. The assumption $\mEt[R_1]>\mEt[R_2]$ then implies $h_1(0)<h_{2}(0)$. Combined with $\lambda(T-)\ne0$, we obtain
$|a_1(T-)|<|a_2(T-)|$.
\subsection{Proof of Proposition \ref{proposition5.3}}\label{sec:proof:croos}

\paragraph{Step 1:} Fix a right-continuous and deterministic function $\lambda(\cdot) \in L^2([0,T))$.
We show that there exists at most one $t^*\in(0,T)$ such that $|a_1(t^*)|=|a_2(t^*|$.  Indeed, let $t^* \in (0,T)$ be arbitrary and satisfy $|a_1(t^*)|=|a_2(t^*|$.  Let $\mathcal{D}(t) \triangleq Q_1(v_1(t)) - Q_2(v_2(t))$. Then, $\mathcal{D}(t^*)=0$.  We now investigate the derivative of $\mathcal{D}(\cdot)$ at $t^*$. 
Recalling the expression for $Q_i(x)$ in \eqref{eq:Xi:Qi}, we have
\begin{align*}
Q_i(x) = \frac{\mEt[\boldsymbol{X}_i e^{-\boldsymbol{X}_i \frac{x}{2}}]}{\mEt[e^{-\boldsymbol{X}_i \frac{x}{2}}]} = \frac{\delta_i(1-p_i)e^{-\delta_i \frac{x}{2}}}{p_i + (1-p_i)e^{-\delta_i \frac{x}{2}}}, \quad x\in[0,\infty).
\end{align*}
Note that
\begin{align*}
Q_i'(x) &= \frac{-\frac{\delta_i^2}{2}(1-p_i)e^{-\delta_i \frac{x}{2}} \left( p_i + (1-p_i)e^{-\delta_i \frac{x}{2}} \right) - \delta_i(1-p_i)e^{-\delta_i \frac{x}{2}} \left( -\frac{\delta_i}{2}(1-p_i)e^{-\delta_i \frac{x}{2}} \right)}{\left( p_i + (1-p_i)e^{-\delta_i \frac{x}{2}} \right)^2} \\
 &= -\frac{\delta_i}{2} \left( \frac{\delta_i(1-p_i)e^{-\delta_i \frac{x}{2}}}{p_i + (1-p_i)e^{-\delta_i \frac{x}{2}}} \right) + \frac{1}{2} \left( \frac{\delta_i(1-p_i)e^{-\delta_i \frac{x}{2}}}{p_i + (1-p_i)e^{-\delta_i \frac{x}{2}}} \right)^2 \\
&= -\frac{\delta_i}{2} Q_i(x) + \frac{1}{2} Q_i(x)^2 \\
&= \frac{1}{2} Q_i(x) (Q_i(x) - \delta_i).
\end{align*}
Differentiating $\mathcal{D}(t)$ and using $v_i'(t) = -|\lambda(t)|^2 (r_0 + Q_i(v_i(t)))^{-2}$ yields
\begin{align*}
\mathcal{D}' &= Q_1'(v_1)v_1' - Q_2'(v_2)v_2' \\
&= \frac{1}{2}Q_1(v_1)(Q_1(v_1)-\delta_1) \left( \frac{-|\lambda|^2}{(r_0+Q_1(v_1))^2} \right) - \frac{1}{2}Q_2(v_2)(Q_2(v_2)-\delta_2) \left( \frac{-|\lambda|^2}{(r_0+Q_2(v_2))^2} \right) \\
&= \frac{|\lambda|^2}{2} \left[ \frac{Q_2(v_2)(Q_2(v_2)-\delta_2)}{(r_0+Q_2(v_2))^2} - \frac{Q_1(v_1)(Q_1(v_1)-\delta_1)}{(r_0+Q_1(v_1))^2} \right].
\end{align*}
For any $t^* \in (0, T)$ such that $\mathcal{D}(t^*) = 0$, we have $Q \triangleq Q_1(v_1(t^*)) = Q_2(v_2(t^*)) > 0$. 
Therefore,
\begin{align*}
\mathcal{D}'(t^*) = \frac{|\lambda(t^*)|^2 Q}{2(r_0+Q)^2} (\delta_1 - \delta_2)>0.
\end{align*}
Thus, there exists a unique $t^*\in(0,T)$ such that $\mathcal{D}(t^*)=0$.

\paragraph{Step 2:}
Define
\begin{align}
    \mathcal{Y}(\Lambda) \triangleq h_1(H_1^{-1}(\Lambda)) - h_2(H_2^{-1}(\Lambda)),\quad \Lambda \in [0, \infty).
\end{align}
We show there exists  a unique $\widehat{\Lambda}\in(0,\infty)$ such that $\mathcal{Y}(\widehat{\Lambda})=0$.  At $\Lambda = 0$, we have $H_i^{-1}(0) = 0$ and $h_i(0) = \dfrac{1}{\mEt[\boldsymbol{R}_i]}$. The assumption $\mEt[\boldsymbol{R}_1] > \mEt[\boldsymbol{R}_2]$ implies $h_1(0) < h_2(0)$, and thus $\mathcal{Y}(0) < 0$. Conversely, as $\Lambda \to \infty$, \eqref{eq:h_inequality} implies that $h_1(H_1^{-1}(\Lambda)) > h_2(H_2^{-1}(\Lambda))$. Therefore, $\mathcal{Y}(\Lambda) > 0$ for sufficiently large $\Lambda$.
By the mean value theorem, there exists at least one $\widehat{\Lambda} \in (0, \infty)$ such that $\mathcal{Y}(\widehat{\Lambda}) = 0$.

We next prove the uniqueness of $\widehat{\Lambda}$.
Suppose, by contradiction, that there exist $0 < \widehat{\Lambda}_1 < \widehat{\Lambda}_2 < \infty$ such that
$\mathcal{Y}(\widehat\Lambda_i) = 0(i=1,2)$.  
Choose a constant market price of risk $\lambda(t) \equiv \lambda_0 \neq 0$ such that
$\Lambda(0) = |\lambda_0|^{2}T > \widehat{\Lambda}_2$.  
Then there exist $0 < t_1 < t_2 < T$ such that $\Lambda(t_i) = \widehat{\Lambda}_i$.
Consequently, $\mathcal{D}(t_1) = \mathcal{D}(t_2) = 0$, which contradicts the conclusion in Step~1.
Therefore, $\widehat{\Lambda}$ is unique.

\paragraph{Step 3:}
We now apply the conclusions in Steps 1 and 2 to prove the three assertions:
\paragraph{(1)}
First, assume $\displaystyle\int_0^T |\lambda(t)|^2 \md t < \widehat{\Lambda}$. We have $\Lambda(t) \in [0, \Lambda(0)] \subset [0, \widehat{\Lambda})$ for all $t \in [0, T)$. Within this interval, $\mathcal{Y}(\Lambda(t)) < 0$, which implies $h_1(v_1(t)) < h_2(v_2(t))$. By \eqref{eq:equilibrium:g}, $|a_1(t)| < |a_2(t)|$ for all $t \in [0, T)$.

\paragraph{(2)}
Next, assume $\displaystyle\int_0^T |\lambda(t)|^2 \md t > \widehat{\Lambda}$. By Proposition \ref{proposition5.2}, we have $|a_1(T-)| < |a_2(T-)|$ . 
Moreover, because $\Lambda(0) > \widehat{\Lambda}$, it follows that $\mathcal{Y}(\Lambda(0)) > 0$, and hence $|a_1(0)| > |a_2(0)|$. Combining this with the result of Step~1, we conclude that there exists a unique time $t^* \in (0,T)$ such that $|a_1(t^*)|=|a_2(t^*)|$. The remaining conclusions then follow immediately.

\paragraph{(3)}
Finally, assume $\displaystyle\int_0^T |\lambda(t)|^2 \md t = \widehat{\Lambda}$.
Then $\Lambda(0) = \widehat{\Lambda}$, and hence $\mathcal{Y}(\Lambda(0)) = 0$, which implies $|a_1(0)| = |a_2(0)|$.
For any $t \in (0, T)$, we have $\Lambda(t) < \widehat{\Lambda}$.
Thus $\mathcal{Y}(\Lambda(t)) < 0$, and hence $|a_1(t)| < |a_2(t)|$ for all $t \in (0, T)$.
\subsection{Proof of Proposition \ref{proposition5.5}}\label{append:t^*P}

Throughout the proof, for $i=1,2$ and a fixed $p\in I$, we write $a_{i,p}$, $Q_{i,p}$, etc., to denote the corresponding functions in Propositions \ref{proposition5.2} and \ref{proposition5.3} corresponding to a given $p\in I$.

We first show that $I$ is an open set. Indeed, define $\mathcal{D}(t, p) \triangleq Q_{1,p}(v_{1,p}(t)) - Q_{2,p}(v_{2,p}(t))$. The crossing time $t^*(p)$ is characterized by $\mathcal{D}(t^*(p), p) = 0$. From the proof of Proposition \ref{proposition5.3}, we know that $\frac{\partial \mathcal{D}}{\partial t}(t^*(p), p) > 0$. Hence, by the implicit function theorem, for each $p\in I$ there exists an open neighborhood in $(0,1)$ on which $t^*(\cdot)$ is uniquely defined and continuously differentiable. This implies that $I$ is an open set.

Recall that $Q_{i,p}(x) = \dfrac{\delta_i (1-p) e^{-\delta_i \frac{x}{2}}}{p + (1-p) e^{-\delta_i \frac{x}{2}}}$, which is strictly decreasing in $x$. Define $U_i(p) \triangleq Q_{i,p}(0) =\delta_i(1-p)$ and $G_{i}(q) \triangleq \dfrac{(r_0+q)^2}{q(\delta_i-q)}$ for $p\in I$ and $q\in (0,\delta_i)$. Then
\begin{align}\label{eq:t^*P}
    \Lambda(t) = \int_0^{v_{i,p}(t)} (r_0 + Q_{i,p}(x))^2 \, \md x = 2 \int_{Q_{i,p}(v_{i,p}(t))}^{U_i(p)} G_{i}(q) \, \md q,
\end{align}
where the second equality arises from the change of variable $q = Q_{i,p}(x)$, which yields $\md x = \dfrac{2}{q(q-\delta_i)} \, \md q$.
For $p\in I$, set $\Lambda^*(p) \triangleq \Lambda(t^*(p))$ and $Q(p) \triangleq Q_{1,p}(v_{1,p}(t^*(p))) = Q_{2,p}(v_{2,p}(t^*(p)))$. Equation \eqref{eq:t^*P} then implies
\begin{align}\label{eq:Psi}
    \Lambda^*(p) = 2 \int_{Q(p)}^{U_1(p)} G_1(q) \, \md q = 2 \int_{Q(p)}^{U_2(p)} G_2(q) \, \md q.
\end{align}
Differentiating \eqref{eq:Psi} w.r.t. $p$ and using $U_i'(p)=-\delta_i$, we obtain
\begin{align}\label{eq:Q'}
    \dfrac{1}{2} (\Lambda^*)'(p) = -\delta_1 G_1(U_1(p)) - G_1(Q(p)) Q'(p)=-\delta_2 G_2(U_2(p)) - G_2(Q(p)) Q'(p).
\end{align}
Because $\delta_1 - Q(p) > \delta_2 - Q(p) > 0$, it follows that $G_1(Q(p)) < G_2(Q(p))$. 
Solving the second equality in \eqref{eq:Q'} for $Q'(p)$ and substituting into the first yields
\begin{align}\label{differential}
    \dfrac{1}{2} (\Lambda^*)'(p) = \dfrac{\delta_2 G_2(U_2(p)) G_1(Q(p)) -\delta_1 G_1(U_1(p)) G_2(Q(p))}{G_2(Q(p)) - G_1(Q(p))}.
\end{align}
Next, a direct computation shows that
\begin{align*}
    \dfrac{\delta_1 G_1(U_1(p))}{G_1(Q(p))} &= \dfrac{(r_0 + \delta_1(1-p))^2}{ p(1-p)\delta_1} \cdot \dfrac{Q(p)(\delta_1-Q(p))}{(r_0+Q(p))^2} \\
    &= \dfrac{Q(p)}{p(1-p)(r_0+Q(p))^2} (r_0 + \delta_1(1-p))^2 \left(1-\dfrac{Q(p)}{\delta_1}\right).
\end{align*}
Similarly,
\begin{align*}
    \dfrac{\delta_2 G_2(U_2(p))}{G_2(Q(p))} &= \dfrac{Q(p)}{p(1-p)(r_0+Q(p))^2} (r_0 + \delta_2(1-p))^2 \left(1-\dfrac{Q(p)}{\delta_2}\right).
\end{align*}
Because $\delta_1 > \delta_2$, it holds that $(r_0 + \delta_1(1-p))^2 > (r_0 + \delta_2(1-p))^2$ and $1 - \dfrac{Q(p)}{\delta_1} > 1 - \dfrac{Q(p)}{\delta_2}$. Therefore, $\dfrac{\delta_1 G_1(U_1(p))}{G_1(Q(p))} > \dfrac{\delta_2 G_2(U_2(p))}{G_2(Q(p))}$. 
Consequently, $(\Lambda^*)'(p) < 0$.  Recall that $\Lambda^*(p) = \Lambda(t^*(p))$. Differentiating both sides w.r.t. $p$ yields
\begin{align*}
    (\Lambda^*)'(p) = \Lambda'(t^*(p)) \frac{\md t^*}{\md p}(p) = -|\lambda(t^*(p))|^2 \frac{\md t^*}{\md p}(p),
\end{align*}
Therefore,
\begin{align*}
    \frac{\md t^*}{\md p}(p) = -\frac{(\Lambda^*)'(p)}{|\lambda(t^*(p))|^2} > 0 \quad \text{for all } p \in I.
\end{align*}

\subsection{Proof of Proposition \ref{proposition5.6}}\label{append:t_1*p}

We retain the notation from the proof of Proposition \ref{proposition5.5}.
Define
$\mathcal{D}(t, p_1, p_2) \triangleq Q_{1, p_1}(v_{1, p_1}(t)) - Q_{2, p_2}(v_{2, p_2}(t))$. By definition, the crossing time $t^*(p_1,p_2)$ satisfies $\mathcal{D}(t^*(p_1,p_2), p_1, p_2) = 0$.
From the proof of Proposition \ref{proposition5.3}, we know that $\frac{\partial \mathcal{D}}{\partial t}(t^{*}(p_1,p_2),p_1,p_2) > 0$.
Therefore, by the implicit function theorem, for any $(p_1,p_2)\in O$, there exists an open neighborhood in $(0,1)^2$ on which $t^*(p_1,p_2)$ is uniquely defined and continuously differentiable. In particular, $O$ is an open set and $t^*$ is differentiable in $O$.

For any $(p_1, p_2) \in O$, define $\Lambda^*(p_1, p_2) \triangleq \Lambda(t^*(p_1, p_2))$ and  $Q(p_1, p_2) \triangleq Q_{1,p_1}(v_{1,p_1}(t^*(p_1,p_2))) = Q_{2,p_2}(v_{2,p_2}(t^*(p_1,p_2)))$.
Analogous to \eqref{eq:Psi}, we obtain
\begin{align}\label{eq:Psi'}
    \Lambda^*(p_1, p_2) = 2\int_{Q(p_1, p_2)}^{U_1(p_1)} G_1(q) \, \md q = 2\int_{Q(p_1, p_2)}^{U_2(p_2)} G_2(q) \, \md q.
\end{align}
Differentiating \eqref{eq:Psi'} w.r.t. $p_1$, and using $\dfrac{\partial U_1}{\partial p_1} = -\delta_1$ and $\dfrac{\partial U_2}{\partial p_1} = 0$, we obtain

\begin{align}
    \dfrac{1}{2} \dfrac{\partial\Lambda^*}{\partial p_1} = -\delta_1 G_1(U_1(p_1)) - G_1(Q(p_1, p_2)) \dfrac{\partial Q(p_1, p_2)}{\partial p_1}= - G_2(Q(p_1, p_2)) \dfrac{\partial Q(p_1, p_2)}{\partial p_1}. \label{eq:sys1_b}
\end{align}
 As shown in the proof of Proposition \ref{proposition5.5}, $\delta_1 > \delta_2$ implies $G_1(Q(p_1, p_2)) < G_2(Q(p_1, p_2))$. 
Solving for $\displaystyle\frac{\partial Q(p_1, p_2)}{\partial p_1}$ from the second equality of \eqref{eq:sys1_b} and  substituting into the first yields
\begin{align*}
     \dfrac{1}{2}\dfrac{\partial\Lambda^*}{\partial p_1} = \dfrac{- \delta_1 G_1(U_1(p_1)) G_2(Q(p_1, p_2))}{G_2(Q(p_1, p_2)) - G_1(Q(p_1, p_2))}<0.
\end{align*}
 Applying the Chain Rule to the identity $\Lambda^*(p_1, p_2) = \Lambda(t^*(p_1, p_2))$ yields
\begin{align*}
    \dfrac{\partial \Lambda^*}{\partial p_1} = \Lambda'(t^*) \dfrac{\partial t^*}{\partial p_1} = -|\lambda(t^*)|^2 \dfrac{\partial t^*}{\partial p_1}.
\end{align*}
Consequently, $\frac{\partial t^*}{\partial p_1}>0$.
Similarly, we obtain $\frac{\partial t^*}{\partial p_2}<0$.

\subsection{Proof of Proposition \ref{proposition5.8}}\label{Append:rh}
From Lemma \ref{proposition5.7}, we know that $h_1(x) < h_2(x)$ for all  $x\ge 0$. it follows that $H_1(y) > H_2(y)$ for all $y > 0$, and in particular, $H_1(\infty) \ge H_2(\infty)>\Lambda(0)$.  This guarantees that the deterministic equilibrium strategy associated with each ${\boldsymbol{R}_i}$ exists and is unique. Moreover, the inverse functions of $H_i$ ($i=1, 2$) satisfy $H_1^{-1}(z) < H_2^{-1}(z)$ for all $z > 0$. From Theorem \ref{theorem3.3}(1), we have $v_i(t) = H_i^{-1}(\Lambda(t))$ ($i=1, 2$) . Consequently,  the monotonicity of $H_i^{-1}$ ($i=1, 2$) implies $v_1(t) \leq v_2(t)$. By (\ref{eq:equilibrium:g}), it holds that $|a_i(t)| = h_i(v_i(t)) |\lambda(t)|$ ($i=1, 2$). Using Lemma \ref{proposition5.7}, we further obtain $h_1(v_1(t)) \leq h_2(v_1(t))\le h_2(v_2(t))$. Therefore,  $|a_1(t)| \le |a_2(t)|$. If $\lambda(t) \neq 0$, then $v_1(t) > 0$ and  $h_1(v_1(t)) < h_2(v_1(t)) \le h_2(v_2(t))$, which yields the strict inequality: $|a_1(t)| < |a_2(t)|$. 

\subsection{Proof of Proposition \ref{proposition5.10}}\label{append:convex}
Let $h(x)$ and $v(t)$ (resp.\ $h_i(x)$ and $v_i(t)$) denote the functions associated with $\boldsymbol{R}$ (resp.\ $\boldsymbol{R}_i$).
From the definition of $h$ and the fact that $\boldsymbol{R}_1,\cdots,\boldsymbol{R}_n$ are i.i.d., we have
\begin{align}
&\frac{1}{h(x)} = \frac{\mEt[\boldsymbol{R} e^{-\frac{1}{2}\boldsymbol{R} x}]}{\mEt[e^{-\frac{1}{2}\boldsymbol{R} x}]} = \frac{\mEt\left[ \displaystyle\left(\sum_{i=1}^n w_i \boldsymbol{R}_i\right) \prod_{k=1}^n e^{-\frac{1}{2} w_k \boldsymbol{R}_k x} \right]}{\displaystyle\mEt\left[ \prod_{k=1}^n e^{-\frac{1}{2} w_k \boldsymbol{R}_k x} \right]}\nonumber\\
=&\frac{\displaystyle\sum_{i=1}^n w_i \mEt\left[ \boldsymbol{R}_i e^{-\frac{1}{2} w_i \boldsymbol{R}_i x} \prod_{k \neq i} e^{-\frac{1}{2} w_k \boldsymbol{R}_k x} \right]}{\displaystyle\prod_{k=1}^n \mEt\left[ e^{-\frac{1}{2} w_k \boldsymbol{R}_k x} \right]}=\sum_{i=1}^n w_i \frac{\mEt[\boldsymbol{R}_i e^{-\frac{1}{2}w_i \boldsymbol{R}_i x}]}{\mEt[e^{-\frac{1}{2}w_i \boldsymbol{R}_i x}]}=\sum_{i=1}^n w_i \frac{1}{h_1(w_i x)}.\label{h:sum}
\end{align}
By Lemma \ref{h:mono}, we have 
$\dfrac{1}{h_1(w_i x)} \geq \dfrac{1}{h_1(x)}$ for any $w_i\in(0,1)$ and $x\geq 0$,  which combined with \eqref{h:sum} yields
$h(x) \leq h_1(x)$ and $H(\infty)\ge H_1(\infty)>\Lambda(0)$.  Consequently, the equilibrium strategy associated with $\boldsymbol{R}$ also exists and is unique.
Analogous to the argument in Proposition \ref{proposition5.8}, we further obtain 
$v(t) \leq v_1(t)$. Because $h(x)$ is increasing in $x$, it follows that
$$ |a(t)| = h(v(t))|\lambda(t)| \leq h(v_1(t))|\lambda(t)| \leq h_1(v_1(t))|\lambda(t)| = |a_1(t)|. $$
If $\boldsymbol{R}_1$ is non-degenerate, then by Lemma~\ref{h:mono}, $h_1$ is strictly increasing, which in turn yields the strict inequality $h(x)<h_1(x)$ for  any $x>0$. Combined with the condition $\lambda(t)\neq 0$, we obtain $|a(t)| < |a_1(t)|$.

\subsection{Proof of Proposition \ref{prop:lln}}\label{proof:prop:lln}
Let $\tilde h_n(\cdot)$ and $\tilde v_n(\cdot)$ be the functions associated with the risk aversion $\boldsymbol{Y}_n$.
From \eqref{h:sum}, we have
\begin{align*}
\tilde h_n(x) = h_1\left(\frac{x}{n}\right)\leq h_1(x).
\end{align*}
Then we have $H(\infty)\ge H_1(\infty)>\Lambda(0)$ and the deterministic equilibrium strategy of $\boldsymbol{Y}_n$ exists and is unique. Moreover, by Lemma~\ref{proposition5.7}, for any fixed $x\ge 0$, the sequence $\{\tilde h_n(x)\}_{n\ge1}$ is decreasing in $n$. Because $\tilde H_n(y) = \displaystyle\int_0^y \tilde h_n^{-2}(z) dz$ and $\tilde v_n(t) = \tilde H_n^{-1}(\Lambda(t))$, the monotonicity of $\tilde h_n$ implies that $\tilde v_n(t)$ is  decreasing in $n$. Consequently, for all $t\in[0,T)$,
\begin{align*}
\tilde h_{n+1}(\tilde v_{n+1}(t)) \leq \tilde h_n(\tilde v_{n+1}(t)) \leq \tilde h_n(\tilde v_n(t)),
\end{align*}
which yields
\begin{align*}
|\tilde a_{n+1}(t)| \le |\tilde a_n(t)| \quad \text{ for all }t \in [0, T).
\end{align*}
By the strong law of large numbers, $\boldsymbol{Y}_n\to\mu$ almost surely as $n\to\infty$. Hence, $\boldsymbol{Y}_n$ converges in distribution to the degenerate random variable $\boldsymbol{Y}\equiv\mu$. Furthermore, the assumptions of Theorem~\ref{thm:converge:distribution} are satisfied because $\mEt[\boldsymbol{Y}_n]=\mu$ for all $n$. We therefore conclude that
\begin{align*}
\lim_{n \to \infty} |\tilde a_n(t)| = \frac{|\lambda(t)|}{\mu}.
\end{align*} 
Combining this limit with the monotonicity established above yields \eqref{conv:an}.
\end{appendices}

\renewcommand{\baselinestretch}{1}\normalsize
\bibliographystyle{abbrvnat}
\bibliography{sample}

@article{chen2025equilibrium,
	author = {Chen, Xiaochen and Guan, Guohui and Liang, Zongxia},
	journal = {Mathematics of Operations Research},
	title = {Equilibrium Portfolio Selection Under Beliefs-Dependent Utilities},
	year = {2025}}

@book{follmer2016stochastic,
	address = {Berlin},
	author = {H. F\"{o}llmer and A. Schied},
	edition = {4th},
	publisher = {De Gruyter},
	title = {Stochastic Finance: An Introduction in Discrete Time},
	year = {2016}}

@conference{yan2019time,
	address = {Cham},
	author = {Yan, Wei and Yong, Jiongmin},
	booktitle = {Modeling, Stochastic Control, Optimization, and Applications},
	editor = {Yin, George and Zhang, Qing},
	pages = {533--569},
	publisher = {Springer International Publishing},
	title = {Time-Inconsistent Optimal Control Problems and Related Issues},
	year = {2019}}

@book{Arrow1970,
	address = {London},
	author = {Arrow, Kenneth Joseph},
	publisher = {North--Holland},
	title = {Essays in the Theory of Risk Bearing},
	year = {1970}}

@article{Pratt1964:RiskAversion,
	author = {Pratt, John W},
	journal = {Econometrica},
	number = {1/2},
	pages = {122-136},
	title = {Risk aversion in the small and in the large},
	volume = 32,
	year = {1964}}

@book{SS2007,
	address = {New York},
	author = {Shaked, Moshe and Shanthikumar, J. George},
	publisher = {Springer},
	title = {Stochastic Orders},
	year = {2007}}

@article{Borell2007,
	author = {Borell, Christer},
	journal = {Mathematical Finance},
	number = {1},
	pages = {143-153},
	title = {MONOTONICITY PROPERTIES OF OPTIMAL INVESTMENT STRATEGIES FOR LOG-{B}ROWNIAN ASSET PRICES},
	volume = {17},
	year = {2007}}

@article{hjz2021,
	author = {Hu, Ying and Jin, Hanqing and Zhou, Xun Yu},
	journal = {Mathematical Finance},
	pages = {1056-1095},
	title = {Consistent investment of sophisticated rank-dependent utility agents in continuous time},
	volume = {31},
	year = {2021}}

@article{hz2020mf,
	author = {Huang, Yu-Jui and Zhou, Zhou},
	journal = {Mathematical Finance},
	number = {3},
	pages = {1103-1134},
	title = {Optimal equilibria for time-inconsistent stopping problems in continuous time},
	volume = {30},
	year = {2020}}

@article{hz2019,
	author = {Huang, Yu-Jui and Zhou, Zhou},
	journal = {SIAM Journal on Control and Optimization},
	number = {1},
	pages = {590-609},
	title = {The Optimal Equilibrium for Time-Inconsistent Stopping Problems---The Discrete-Time Case},
	volume = {57},
	year = {2019}}

@article{hw2021,
	author = {Huang, Yu-Jui and Wang, Zhenhua},
	journal = {SIAM Journal on Control and Optimization},
	number = {2},
	pages = {1705-1729},
	title = {Optimal Equilibria for Multidimensional Time-Inconsistent Stopping Problems},
	volume = {59},
	year = {2021}}

@book{Feller1971,
	address = {New York},
	author = {Feller, William},
	edition = {2nd},
	publisher = {John Wiley \& Sons},
	title = {An Introduction to Probability Theory and Its Applications},
	volume = {2},
	year = {1971}}

@article{Pollard1946,
	author = {Harry Pollard},
	journal = {Bulletin of the American Mathematical Society},
	number = {10},
	pages = {908 -- 910},
	publisher = {American Mathematical Society},
	title = {{The representation of $e^{ - x^\lambda }$ as a Laplace integral}},
	volume = {52},
	year = {1946}}

@article{liang2024dynamic,
	author = {Liang, Zongxia and Wang, Sheng and Xia, Jianming and Yuan, Fengyi},
	journal = {arXiv preprint arXiv:2401.08323},
	title = {Dynamic portfolio selection under generalized disappointment aversion},
	year = {2024}}

@article{Hu2017,
	author = {Hu, Ying and Jin, Hanqing and Zhou, Xun Yu},
	journal = {SIAM Journal on Control and Optimization},
	number = {2},
	pages = {1261-1279},
	title = {Time-Inconsistent Stochastic Linear-Quadratic Control: Characterization and Uniqueness of Equilibrium},
	volume = {55},
	year = {2017}}

@article{Merton1971,
	author = {Robert C Merton},
	journal = {Journal of Economic Theory},
	number = {4},
	pages = {373-413},
	title = {Optimum consumption and portfolio rules in a continuous-time model},
	volume = {3},
	year = {1971}}

@article{Strotz1955,
	author = {R. H. Strotz},
	journal = {The Review of Economic Studies},
	number = {3},
	pages = {165--180},
	publisher = {[Oxford University Press, Review of Economic Studies, Ltd.]},
	title = {Myopia and Inconsistency in Dynamic Utility Maximization},
	urldate = {2023-11-08},
	volume = {23},
	year = {1955}}

@article{Bjork2017,
	author = {Bj{\"o}rk, Tomas and Khapko, Mariana and Murgoci, Agatha},
	journal = {Finance and Stochastics},
	number = {2},
	pages = {331-360},
	publisher = {Springer},
	title = {On time-inconsistent stochastic control in continuous time},
	volume = {21},
	year = {2017}}

@article{Desmettre2023,
	author = {Desmettre, Sascha and Steffensen, Mogens},
	journal = {Mathematical Finance},
	number = {3},
	pages = {946-975},
	title = {Equilibrium investment with random risk aversion},
	volume = {33},
	year = {2023}}

@article{Hernandez2023,
	author = {Hern{\'{a}}ndez, Camilo and Possama{\"{i}}, Dylan},
	journal = {Annals of Applied Probability},
	number = {2},
	pages = {1196--1258},
	title = {{Me, Myself and I: A general theory of non-Markovian time-inconsistent stochastic control for sophisticated agents}},
	volume = {33},
	year = {2023}}

@article{ekeland2006being,
	author = {Ekeland, Ivar and Lazrak, Ali},
	journal = {arXiv preprint math/0604264},
	title = {Being serious about non-commitment: {S}ubgame perfect equilibrium in continuous time},
	year = {2006}}

@article{hu2012time,
	author = {Hu, Ying and Jin, Hanqing and Zhou, Xun Yu},
	journal = {SIAM Journal on Control and Optimization},
	number = {3},
	pages = {1548--1572},
	publisher = {SIAM},
	title = {Time-inconsistent stochastic linear--quadratic control},
	volume = {50},
	year = {2012}}

@article{he2021equilibrium,
	author = {He, Xue Dong and Jiang, Zhao Li},
	journal = {SIAM Journal on Control and Optimization},
	number = {5},
	pages = {3860--3886},
	publisher = {SIAM},
	title = {On the equilibrium strategies for time-inconsistent problems in continuous time},
	volume = {59},
	year = {2021}}

@article{chetty2006new,
	author = {Chetty, Raj},
	journal = {American Economic Review},
	number = {5},
	pages = {1821--1834},
	publisher = {American Economic Association},
	title = {A new method of estimating risk aversion},
	volume = {96},
	year = {2006}}

@article{liang2025short,
	author = {Liang, Zongxia and Wang, Sheng and Xia, Jianming},
	journal = {SIAM Journal on Financial Mathematics},
	number = {1},
	pages = {SC12--SC23},
	title = {Short Communication: An Integral Equation in Portfolio Selection with Time-Inconsistent Preferences},
	volume = {16},
	year = {2025}}

@article{merton1969lifetime,
	author = {Merton, Robert C},
	journal = {The Review of Economics and Statistics},
	number = {3},
	pages = {247--257},
	publisher = {JSTOR},
	title = {Lifetime portfolio selection under uncertainty: The continuous-time case},
	volume = {51},
	year = {1969}}

@article{samuelson1969lifetime,
	author = {Samuelson, Paul A},
	journal = {The Review of Economics and Statistics},
	number = {3},
	pages = {239--246},
	publisher = {JSTOR},
	title = {Lifetime Portfolio Selection By Dynamic Stochastic Programming},
	volume = {51},
	year = {1969}}

@article{guiso2018time,
	author = {Guiso, Luigi and Sapienza, Paola and Zingales, Luigi},
	journal = {Journal of Financial Economics},
	number = {3},
	pages = {403--421},
	publisher = {Elsevier},
	title = {Time varying risk aversion},
	volume = {128},
	year = {2018}}

@article{gordon2000preference,
	author = {Gordon, Stephen and St-Amour, Pascal},
	journal = {American Economic Review},
	number = {4},
	pages = {1019--1033},
	publisher = {American Economic Association},
	title = {A preference regime model of bull and bear markets},
	volume = {90},
	year = {2000}}

@article{wei2026time,
	author = {Wei, Jiaqin and Xia, Jianming and Zhao, Qian},
	journal = {Mathematics of Operations Research},
	title = {Time-Consistent Portfolio Selection for Rank-Dependent Utilities in a Constrained Market},
	volume = {forthcoming},
	year = {2026}}

@article{WANG2025103140,
	author = {Ling Wang and Bowen Jia},
	journal = {Insurance: Mathematics and Economics},
	pages = {103140},
	title = {Equilibrium investment strategies for a defined contribution pension plan with random risk aversion},
	volume = {125},
	year = {2025}}

@article{he2025dynamic,
	author = {He, Xue Dong and Jiang, Zhaoli and Xia, Jianming},
	journal = {Available at SSRN 5449757},
	title = {Dynamic Portfolio Selection under Monotone Additive Statistics in the {H}eston Model},
	year = {2025}}

@article{bjork2010general,
	author = {Bjork, Tomas and Murgoci, Agatha},
	journal = {Available at SSRN 1694759},
	title = {A general theory of {Markovian} time inconsistent stochastic control problems},
	year = {2010}}

@article{xia2024optimal,
	author = {Xia, Jianming},
	journal = {SIAM Journal on Financial Mathematics},
	number = {1},
	pages = {54--92},
	publisher = {SIAM},
	title = {Optimal investment with risk controlled by weighted entropic risk measures},
	volume = {15},
	year = {2024}}

@article{caperaa1988tail,
	author = {Cap{\'e}ra{\`a}, Philippe},
	journal = {The Annals of Statistics},
	number = {1},
	pages = {470--478},
	publisher = {JSTOR},
	title = {Tail ordering and asymptotic efficiency of rank tests},
	volume = {16},
	year = {1988}}

@article{xia2011risk,
	author = {Xia, Jianming},
	journal = {SIAM Journal on Control and Optimization},
	number = {5},
	pages = {1916--1937},
	publisher = {SIAM},
	title = {Risk aversion and portfolio selection in a continuous-time model},
	volume = {49},
	year = {2011}}

@book{rudin1976principles,
	author = {Rudin, Walter},
	edition = {3rd},
	publisher = {McGraw-Hill},
	title = {Principles of Mathematical Analysis},
	year = {1976}}

@article{liang2025dynamic,
	author = {Liang, Zongxia and Xia, Jianming and Yuan, Fengyi},
	journal = {Mathematics of Operations Research},
	publisher = {INFORMS},
	title = {Dynamic portfolio selection for nonlinear law-dependent preferences},
	year = {2025}}

@article{aquino2025equilibrium,
	author = {Aquino, Luca De Gennaro and Desmettre, Sascha and Havrylenko, Yevhen and Steffensen, Mogens},
	journal = {arXiv preprint arXiv:2512.21149},
	title = {Equilibrium investment under dynamic preference uncertainty},
	year = {2025}}

\end{document}